\documentclass[aps,pra,notitlepage,amsmath,amssymb,11pt,nofootinbib]{revtex4-1}

\usepackage{graphicx,epsfig,epsf}
\usepackage{dcolumn}
\usepackage{bm}
\usepackage{amsthm}
\usepackage{hyperref}

\begin{document}
\newcommand{\ri}{{\rm i}}
\newcommand{\re}{{\rm e}}
\newcommand{\bx}{{\bf x}}
\newcommand{\bd}{{\bf d}}
\newcommand{\br}{{\bf r}}
\newcommand{\bt}{{\mathbf{t}}}
\newcommand{\bk}{{\bf k}}
\newcommand{\bE}{{\bf E}}
\newcommand{\bR}{{\bf R}}
\newcommand{\bM}{{\bf M}}
\newcommand{\bn}{{\bf n}}
\newcommand{\bs}{{\bf s}}
\newcommand{\bu}{{\mathbf{u}}}
\newcommand{\tr}{{\rm tr}}
\newcommand{\tbs}{\tilde{\bf s}}
\newcommand{\rSi}{{\rm Si}}
\newcommand{\beps}{\mbox{\boldmath{$\epsilon$}}}
\newcommand{\bsigma}{\mbox{\boldmath{$\sigma$}}}
\newcommand{\bthe}{\mbox{\boldmath{$\theta$}}}
\newcommand{\brho}{\mbox{\boldmath{$\rho$}}}
\newcommand{\rg}{{\rm g}}
\newcommand{\xmax}{x_{\rm max}}
\newcommand{\ra}{{\rm a}}
\newcommand{\rx}{{\rm x}}
\newcommand{\rs}{{\rm s}}
\newcommand{\rP}{{\rm P}}
\newcommand{\up}{\uparrow}
\newcommand{\down}{\downarrow}
\newcommand{\hc}{H_{\rm cond}}
\newcommand{\kb}{k_{\rm B}}
\newcommand{\cI}{{\cal I}}
\newcommand{\tit}{\tilde{t}}
\newcommand{\cE}{{\cal E}}
\newcommand{\cC}{{\cal C}}
\newcommand{\Ubs}{U_{\rm BS}}
\newcommand{\qq}{{\bf ???}}
\newcommand*{\etal}{\textit{et al.}}
\newcommand{\pf}{{\rm PF}}

\newcommand{\be}{\begin{equation}}
\newcommand{\ee}{\end{equation}}
\newcommand{\bfg}{\begin{figure}}
\newcommand{\efg}{\end{figure}}
\newcommand{\bra}{\langle}
\newcommand{\ket}{\rangle}
\newcommand{\Itwo}{\mathbb{1}_2}
\newcommand{\I}{\mathcal{I}}
\newcommand{\al}{\alpha}

\newcommand{\C}{\mathbb{C}}
\newcommand{\M}{\mathcal{M}}
\newcommand{\e}{\varepsilon}

\newtheorem{theorem}{Theorem}[section]
\newtheorem{definition}[theorem]{Definition}
\newtheorem{proposition}[theorem]{Proposition}
\newtheorem{corollary}[theorem]{Corollary}
\newtheorem{lemma}[theorem]{Lemma}
\newtheorem{conjecture}[theorem]{Conjecture}

\sloppy

\title{A universal set of qubit quantum channels} 
\author{Daniel Braun$^{1}$, Olivier Giraud$^{2}$, Ion Nechita$^{1}$, Cl\'ement Pellegrini$^{3}$, Marko \v Znidari\v c$^{4}$}
\affiliation{$^1$ Laboratoire de Physique Th\'{e}orique, Universit\'{e} Paul
  Sabatier and CNRS, 118, route de Narbonne, 31062 Toulouse, France} 
\affiliation{$^2$ Universit\'e Paris-Sud, CNRS, LPTMS, UMR8626, 91405 Orsay,
France}
\affiliation{$^3$ Institut de Math\'ematiques de Toulouse, Universit\'e Paul Sabatier,
118, route de Narbonne, 31062 Toulouse, France}
\affiliation{$^4$ Department of Physics, Faculty of Mathematics and Physics, University of Ljubljana, Ljubljana, Slovenia}

\begin{abstract}
We investigate the set of quantum channels  acting on a
single qubit.  We provide an alternative, compact generalization of the
Fujiwara-Algoet conditions for complete positivity to non-unital qubit channels, which we then use
to characterize the possible geometric forms of the pure output of the channel. 
We provide 
universal sets of quantum channels for all unital qubit channels as well as
for all extremal (not necessarily unital) qubit channels, in the sense
that all qubit channels in these sets can be obtained by concatenation
of channels in the corresponding universal set. We also show
that our universal sets are essentially minimal.
\end{abstract}

\maketitle
\tableofcontents

\section{Introduction}
The science of quantum information has attracted tremendous interest over
the last twenty years after it was realized that certain computational tasks
might be solvable more efficiently on a quantum computer than on a classical
computer  \cite{Nielsen00}. For
example, a quantum algorithm exists that can factorize a large integer in
a time polynomial in the number of digits, whereas no such classical algorithm
is known \cite{Shor94}. Other examples of quantum speed-up
include algorithms for searching an unsorted database \cite{Grover97}, the
hidden subgroup problem \cite{Simon94}, the approximation of Jones polynomials 
\cite{Aharonov05}, or 
the solution of linear systems of equations \cite{Harrow09.2} with a
recent application to data fitting \cite{Wiebe12}. 

Traditionally, quantum information science was mostly concerned with unitary
time evolution, which is appropriate for well isolated quantum systems.  An
important step in the development of quantum information theory was the
realization that all unitary operations in the exponentially large Hilbert
space of many qubits can be broken down into elementary unitary operations
that act on only one or two qubits at the same time \cite{Deutsch85}.  In
fact, the 
combination of a single fixed entangling unitary gate that can be applied on
any two qubits and the continuous set of all unitary operations on all
single qubits form a ``universal gate set'', which is at the heart of the
circuit 
paradigm of quantum
computing
\cite{Sleator95,Deutsch95,Barenco95b,Barenco95a,DiVincenzo95}.   

However, in reality no quantum system is perfectly isolated from its
environment.  At the very least, the need of state-preparation and the
read-out of the results require interaction with the external world.
Interactions with the environment lead typically to decoherence and destroy
the quantum effects such as interference \cite{Braun06,Arnaud07,Roubert08}
and entanglement which are vital 
for the quantum computational speed-up \cite{Jozsa03}.  But couplings to the
environment can also have beneficial effects.  In particular, it is possible
to create entanglement through purely dissipative processes
\cite{Plenio99,Braun02,Benatti05}.  Dissipative processes may be used
to effectively confine the dynamics to a part of Hilbert space where
decoherence is strongly reduced (``decoherence-free subspaces'', see
\cite{Duan98,Zanardi97,Braun98,Lidar98,PhysRevA.76.012331}) and at the
same time enable entangling 
quantum gates in a simple fashion \cite{Beige00,Beige00b}. 
Indeed, 
non-unitary propagation of quantum states opens up a much larger field of
operations, and it is desirable to achieve a thorough understanding of the
set of these ``quantum channels'', defined quite generally as linear,
completely positive trace preserving maps of a density matrix
\cite{Nielsen00,Zyczkowski04}. Quantum channels have also played an
important 
role  in terms of error models and for the development of quantum
error correction \cite{Shor95,Steane96,Gottesman96}.

Given the importance of universal unitary gate sets for unitary quantum
computation, one might expect that 
simple universal sets of quantum channels might become equally important for
exploring the full power of the most general quantum operations allowed by
nature.  However, the set of quantum channels is much larger and has
more complicated geometry than the set of
unitary operations \cite{Bengtsson06}.  Important early work, long before
the rise of quantum information 
theory,  has provided us with powerful tools that 
allow us to assess the crucial complete positivity of a channel, such as the
Kraus 
decomposition \cite{Kraus83}, the Choi matrix \cite{Choi75}, or the
Lindblad form of Markovian equations of motion
\cite{Kossakowski72,Gorini76,Lindblad76}.  
But the full understanding of completely positive maps, especially with
respect to composition, remains a formidable mathematical problem.
  
Here we will restrict ourselves largely to single qubit channels. 
Ruskai and co-workers \cite{King01,Ruskai02}, and Wolf and
Cirac \cite{Wolf08} have made important contributions which we will heavily
use. Fujiwara and Algoet provided simple inequalities that
characterize the set of all unital qubit channels, i.e.~qubit channels
that map the identity matrix on itself \cite{Fujiwara99}.  We provide
a compact 
generalization of these conditions to the general non-unital case where the
Bloch sphere is mapped to an ellipsoid contained in the Bloch sphere. These new conditions have more natural geometric interpretation than (equivalent) conditions presented in \cite{Ruskai02} and allow us to classify qubit channels in terms of their pure output (PO).  We show in particular that a PO in the form of
a circle of non-zero radius on the Bloch sphere is forbidden by the
requirement of complete positivity.  For unital qubit channels we derive
a universal set of qubit channels in the sense that all unital qubit
channels can be obtained by concatenation of channels from the
universal set.  Furthermore, we provide a set of universal
channels for extremal (but not necessarily unital) qubit
channels. Since all qubit channels can be obtained by simple
convex combination (i.e.~random classical sampling) of extremal ones,
this essentially solves the problem of a universal channel set for a
single qubit. 

Very recently, a different approach to universal families of
qubit quantum channels has been pursued in \cite{wang_solovay-kitaev_2013}. The
authors use universal unitary gates to approximate the unitary appearing in
the Stinespring dilation of a qubit channel (realizing the channel as a unitary evolution on a larger space). Our approach is more intrinsic,
since we do not refer to any dilations of channels. Our universal set contains only single qubit operations, whereas the $\mathrm{CNOT}$ gate, acting on 2 qubits, is used in \cite{wang_solovay-kitaev_2013}. However, the mathematical objects used in \cite{wang_solovay-kitaev_2013} are similar to ours: the superoperator matrix $T$ \eqref{T} and the description of extreme channels, see Proposition \ref{prop:extremal-channels}. 

We begin with the introductory Section II. and III.A, where we recall well-known facts about qubit channels and their geometry. In Sections III.B and III.C we present new results about the minimal set of universal channels for all unital qubit channels. In Section IV.A we derive an alternative necessary and sufficient inequality that characterizes all non-unital qubit channels. In Section IV.B we classify all qubit channels by the number of PO that they have~\footnote{After finishing this work we have learned that the main result in Section IV.B essentially follows from the result on the number of POs of quasi-extreme channels presented in \cite{Ruskai02} (we would like to thank Mary-Beth Ruskai for pointing that out). Our presentation of this result only involves elementary geometric considerations.}. In Section V. we present new results on universal sets of channels for all extremal qubit channels.

\section{Parametrization of qubit channels}

Let $\mathcal M_d(\mathbb C)$ be the set of complex $d\times d$ matrices,
and $\mathcal D_d\subset\mathcal M_d(\mathbb C)$ the set of density
matrices, that is, positive Hermitian $d\times d$ matrices with trace
one. We denote by $\mathcal P_d\subset \mathcal D_d$ the set of pure states
(density matrices of rank one), and by $\mathcal U_d$ the set of all
$d\times d$ unitary matrices. A channel $\Phi : \mathcal M_d(\mathbb C) \to
\mathcal M_d(\mathbb C)$ is a \emph{completely positive} and \emph{trace
  preserving} linear map; % 
in particular, it maps density matrices to density matrices. We denote by
$\mathcal C_d$ the set of channels acting on a $d$-dimensional quantum
system.  

An important class of channels is the class of unitary channels.  These are
linear mappings $\Phi_U: \mathcal M_d(\mathbb C) \to \mathcal M_d(\mathbb
C)$ which map $\rho\in\mathcal M_d(\mathbb C)$ by a unitary conjugation to
$\Phi_U(\rho)=U\rho 
U^\dagger$, where $U\in\mathcal U_d$. 
Slightly abusing notation, we shall write $\mathcal U_d$ for the set of unitary channels. 

It is known that for a $d$-dimensional Hilbert space 
(of, say, $k$ qubits, $d=2^k$), an arbitrary unitary operation can be
obtained by concatenation of the controlled-NOT (CNOT)-gate and the continuous set of
all single qubit unitaries
\cite{Deutsch85,Deutsch95,DiVincenzo95,Barenco95a,Barenco95b}. Our aim is to
find a minimal set such that any channel can be realized by iterated
application of channels from that minimal set. More specifically, we have
the following definition. 

\begin{definition} A set of quantum channels $\mathcal F_d \subset \mathcal
  C_d$ is said to be \emph{universal} 
for a set of channels $\tilde{\mathcal C}_d\subset{\mathcal C}_d$ if for all
channels $\Phi \in \tilde{\mathcal C}_d$  there exist 
channels $\Phi_1, \ldots, 
  \Phi_n \in \mathcal F_d$ such that  
\begin{equation}\label{eq:concatenation}
	\Phi = \Phi_n \circ \Phi_{n-1} \circ \cdots \circ \Phi_2 \circ \Phi_1.
\end{equation}
\end{definition}
The two subsets  $\tilde{\mathcal C}_d\subset {\mathcal C}_d$ that we will
investigate are unital qubit channels and extremal qubit channels, to be
defined below. 
The CNOT gate together with single qubit unitaries form a universal set of
\emph{unitary} quantum channels in the case $\tilde{\mathcal C}=\mathcal
U_d$.

For any set of channels $\mathcal F$, we denote by $\langle \mathcal
F\rangle_n $ the set of channels generated by $n$
concatenations of elements from $\mathcal F$ (as in \eqref{eq:concatenation}), and by $\langle \mathcal
F\rangle$ the set of all channels generated by $\mathcal F$, with no
restriction on the number of concatenations. Our aim is to find a set
$\mathcal F$, as small as possible, such that $\langle \mathcal
F\rangle=\tilde{\mathcal C_d}$.

\subsection{Qubit channels}
\label{sec.quch} From now on we fix $d=2$ 
and we consider qubit channels $\Phi\in\mathcal C_2$. Any state
$\rho\in\mathcal M_2(\mathbb C)$ can be expanded in the basis of Pauli
matrices $\sigma_i$ as $\rho=\frac12\sum_{i=0}^3r_i\sigma_i$, where $r_i\in
\mathbb{R}$ and  
$\sigma_0=I_2$ is the identity matrix in $\mathcal M_2(\mathbb
C)$. Normalization of the trace $\tr\rho=1$ implies $r_0=1$. The components
$r_i$, $i=1,2,3$, form the Bloch vector $\br=(r_1,r_2,r_3)$. Pure states
have $||\br||=1$ and form the Bloch sphere, whereas the set of all other
states inside the Bloch sphere correspond to mixed states ($\tr\rho^2<1$).  

Any linear map $\Phi\in\mathcal C_2$ acting on
$\rho=\frac12I_2+\frac12\br.\bsigma$, with
$\bsigma=(\sigma_1,\sigma_2,\sigma_3)\equiv (\sigma_x,\sigma_y,\sigma_z)$,
can be represented by a real $4\times 4$ matrix $T$ that maps the components
of $r_i$ to new ones,  
\begin{equation} \label{r'r}
(1,r_1',r_2',r_3')^t=T_\Phi\,(1,r_1,r_2,r_3)^t\,.
\end{equation}
The most general linear completely positive map of $\rho$ is then given by \cite{King01,Ruskai02,Bengtsson06}  
\begin{equation}\label{T}
	T_\Phi = 
	\begin{pmatrix}
		1 & 0_{1\times 3} \\
		\bt_\Phi & M_\Phi
	\end{pmatrix}
\end{equation} 
with $0_{1\times 3}=(0,0,0)$, $M_\Phi$ is a real $3\times 3$ matrix, and
$\bt_\Phi\in\mathbb{R}$ a vector. It
induces an affine map $\br'=M_\Phi\br+\bt_\Phi$ on the Bloch
vector. Composition of two channels, $\Phi=\Phi_2\circ\Phi_1$, implies 
\begin{align}
\label{eq:compose-M} M_\Phi&=M_{\Phi_2}M_{\Phi_1}\,, \\
\label{eq:compose-t} \bt_\Phi&=M_{\Phi_2} \bt_{\Phi_1}+\bt_{\Phi_2}. 
\end{align}
In the following we  drop the index $\Phi$ when it is clear what
channel $\bt_\Phi$ and $M_\Phi$ refer to.
Qubit channels can thus be seen as maps acting on Bloch vectors thanks to
the isomorphism  
\begin{equation}
\label{isomorphism}
\mathrm{SU}(2) / \mathbb Z_2 \cong \mathrm{SO}(3).
\end{equation}
In particular, the channel $\Phi_U$ corresponding to unitary conjugation
with
$U=\exp(i\phi\,\bn.\bsigma)=\cos(\phi)I_2+i\sin(\phi)\bn.\bsigma$ is
equivalent to a rotation $R_U\in SO(3)$ of the Bloch vector about axis $\bn$
by an angle $2\phi$. 

Complete positivity of a qubit channel can be characterized by the
positivity of its Choi matrix $C_\Phi$. The \emph{Choi matrix} of a channel
$\Phi$ is 
defined by  
\begin{equation}
\label{choi}
C_\Phi = [\Phi \otimes I_2](|\mathrm{Bell} \rangle \langle
\mathrm{Bell} |),
\end{equation}
where $|\mathrm{Bell}\rangle = \frac{1}{\sqrt 2} ( |00\rangle + |11\rangle)$
is one of the Bell states for two qubits. Alternatively, the Choi matrix can 
be defined by a reshuffling of the indices of the propagator in the 
computational basis \cite{Bengtsson06}. A channel is completely positive, if
and only if the Choi matrix is non-negative \cite{Choi75}.  The eigenvectors of
the Choi matrix yield, after reshaping them to a matrix and multiplication
with the square root of the corresponding eigenvalue, the Kraus operators
$A_i$ of the channel, defined through $\Phi:\rho\mapsto\sum_{i=1}^rA_i\rho
A_i^\dagger$. The  minimal number $r$ of Kraus operators is equal to the
number of non-zero eigenvalues of the Choi matrix and is called the Kraus
rank \cite{Bengtsson06}.

\subsection{Signed singular values for qubit channels}
\label{sec.ssv}
For any matrix $M_\Phi$, there exist two orthogonal matrices $M_1,M_2$ such
 that the singular value decomposition (SVD) of $M_\Phi$ reads
 $M_\Phi=M_1DM_2$, with $D$ a diagonal matrix with non-negative entries. Any
 orthogonal matrix $M$ is such that either $M$ or $-M$ is in $SO(3)$ (in the
 latter case $M$ corresponds to an improper rotation, i.e.~a concatenation
 of a proper rotation with a central inversion). Let $U_i\in SU(2)$ be a
 unitary matrix corresponding to $M_i$ via the isomorphism
 \eqref{isomorphism} if $M_i\in SO(3)$, or corresponding to $-M_i$
 otherwise. Then  $M_\Phi=R_{U_1}\Lambda R_{U_2}$ with $\Lambda=D$ if both
 $M_1$ and $M_2$ are in $SO(3)$ or $\Lambda=-D$ if exactly one of the $M_i$
 is in $SO(3)$. The channel $\Phi$ can thus be decomposed into
 $\Phi=\Phi_{U_1}\circ\Phi_\Lambda\circ\Phi_{U_2}$, where $\Phi_{U_i}$ are
 unitary conjugations and $\Phi_\Lambda$ a channel whose matrix $M=\Lambda$
 is diagonal. We call the diagonal values of $\Lambda={\rm
 diag}(\lambda_1,\lambda_2,\lambda_3)$ the ``signed singular values'' of
 $\Phi$ \cite{Bourdon04}. There is of course arbitrariness in the order in
 which the signed singular values are labeled. Changing the order of the
 $\lambda_i$ just amounts to changing the order in which the eigenvectors of
 $M_\Phi$ appear in matrices $R_{U_i}$. More precisely, for a permutation of
 three elements, say $\sigma=(123)$, we consider the permutation channel
 $\Phi_{\sigma}$ defined by  
\begin{equation}
M_\sigma=\left(\begin{array}{ccc}0&0&1\\1&0&0\\0&1&0\end{array}\right)
\end{equation}
and $\bt_{\sigma}={\bf 0}$. This channel allows to permute the eigenvalues of
a matrix $M$. Namely, if $\Lambda={\rm
  diag}(\lambda_1,\lambda_2,\lambda_3)$, we get $M_\sigma\Lambda
M_\sigma^{\dagger}={\rm diag}(\lambda_3,\lambda_1,\lambda_2)$. 
Moreover, one can simultaneously change the signs of (exactly) two singular
values of $D$ by concatenating with a unitary channel. For example,
concatenation with $e^{i\pi\sigma_z/2}={\rm diag}(-1,-1,1)$ changes the
signs of $\lambda_1$ and $\lambda_2$. 

Summarizing, up to unitary rotations, any qubit channel can be written as 
\begin{equation}\label{eq:T-diagonal}
	T_\Phi = 
	\begin{pmatrix}
		1 & 0 & 0 & 0 \\
		t_1 & \lambda_1 & 0 & 0 \\
		t_2 & 0 & \lambda_2 & 0 \\
		t_3 & 0 & 0 & \lambda_3
	\end{pmatrix},
\end{equation}
where $\bm{\lambda}_\Phi = (\lambda_1, \lambda_2, \lambda_3)$ is the vector of (signed) singular values of the matrix $M$ from \eqref{T} and $\bt_\Phi = (t_1, t_2, t_3)$ are the coordinates (in the Pauli basis) of $\Phi(I_2/2)$.
In the next section we consider the simpler case of unital qubit channels,
for which $\bt_{\Phi}={\bf 0}$. We shall turn to non-unital channels in
Section \ref{nonunital}. 

\section{Universal set of unital qubit channels}

\subsection{Geometry of unital channels}
Unital qubit channels $\Phi$ are defined as channels which leave the fully
mixed  
state $\rho_0 = I_2/2$ invariant. In the representation \eqref{T}, a channel
$\Phi$ is unital if and only if $\bt_\Phi=0$.  
Using $\Phi(\sigma_i) = \lambda_i \sigma_i$ for $i=0,1,2,3$ with
$\lambda_0=1$, we obtain
the Choi matrix \eqref{choi} in the computational basis for unital qubit channels, 
\begin{equation}
C_\Phi = \frac{1}{4}
\begin{pmatrix}
1+\lambda_3 & 0 & 0 & \lambda_1+\lambda_2 \\
0 & 1-\lambda_3 & \lambda_1-\lambda_2 & 0 \\
0 & \lambda_1-\lambda_2 & 1-\lambda_3 & 0 \\
\lambda_1+\lambda_2 & 0 & 0 & 1+\lambda_3
\end{pmatrix}.
\end{equation}
Using the obvious block-structure of $C_\Phi$, its eigenvalues
$q_0,q_1,q_2,q_3$ are easily computed as  
\begin{align}
\label{qi_0}
q_0 &= (1+\lambda_1+\lambda_2+\lambda_3)/4 \\
q_1 &= (1+\lambda_1-\lambda_2-\lambda_3)/4\\
q_2 &= (1-\lambda_1+\lambda_2-\lambda_3)/4\\
q_3 &= (1-\lambda_1-\lambda_2+\lambda_3)/4\,.
\label{qi_3}
\end{align}
According to Choi's theorem \cite{Choi75}, the linear map $\Phi$ is completely positive
iff its Choi matrix $C_\Phi$ is positive, i.e. $q_i \geq 0$, $i=0,1,2,3$. 
These four inequalities are exactly equivalent to the celebrated Fujiwara-Algoet conditions (FAC) for the complete positivity of a unital qubit channel
\cite{Fujiwara99}, 
\begin{equation}\label{FAC}
\left\{ 
\begin{array}{cc}
1 + \lambda_3 &\geq |\lambda_1 + \lambda_2| \\
 1 - \lambda_3 &\geq |\lambda_1 - \lambda_2|.
\end{array}
\right.
\end{equation}
The FAC \eqref{FAC} provide four inequalities; equality in any one of them is equivalent to $q_i=0$ for some $i$. Note that similar conditions were obtained in \cite{Fujiwara99} for a particular subclass of non-unital channels, but we shall address this question in Section \ref{sec:generalized-FAC}.

To each channel of the form \eqref{eq:T-diagonal} one can associate a point
in $\mathbb{R}^3$  specified by its coordinates
$(\lambda_1,\lambda_2,\lambda_3)$. Let $V_1\equiv(1,1,1)$,
$V_2\equiv(1,-1,-1)$, $V_3\equiv(-1,1,-1)$, and $V_4\equiv(-1,-1,1)$ be four
points in $\mathbb{R}^3$. Point $V_1$ corresponds to the identity channel;
points $V_i$, $2\leq i\leq 4$, correspond respectively to deterministic bit
flip, bit-phase flip, and phase flip \cite{Bengtsson06}. The vertices
$V_1,V_2,V_3,V_4$ define a regular tetrahedron $\mathcal T$. Rewriting
relations \eqref{qi_0}-\eqref{qi_3} as 
\begin{align}
	\begin{pmatrix} \lambda_1 \\ \lambda_2 \\ \lambda_3 \end{pmatrix} &= 
	q_0 \begin{pmatrix} 1 \\ 1 \\ 1 \end{pmatrix} +
	q_1 \begin{pmatrix} 1 \\ -1 \\ -1 \end{pmatrix} +
	q_2 \begin{pmatrix} -1 \\ 1 \\ -1 \end{pmatrix} +
	q_3 \begin{pmatrix} -1 \\ -1 \\ 1 \end{pmatrix}  \\
	&=q_0 V_1 + q_1 V_2 + q_2 V_3 + q_3 V_4,
\end{align}
one can interpret the FAC \eqref{FAC} geometrically by saying that the signed singular values $\lambda_i$ of a quantum channel $\Phi$ must lie inside 
the tetrahedron $\mathcal T$. Equality in one of the four inequalities
\eqref{FAC} defines a face of $\mathcal T$.  

Because of the arbitrariness in the ordering of the $\lambda_i$ mentioned in
 the previous section, distinct points of $\mathcal T$ can be exchanged with
 one another via permutation channels. Moreover, rotation by an angle $\pi$
 about any of the three coordinate axes $x$, $y$ or $z$ flips the signs of
 the two coordinates corresponding to the directions perpendicular to the
 rotation axis. Starting from the identity channel represented by $V_1$,
 concatenations with unitary channels allow to reach the channels $\Phi_i$
 represented by points $V_i$. Therefore, the four vertices of $\mathcal T$
 are equivalent up to unitary transformations. Note that an ''inversion
 channel'' with $\bm{\lambda}=-V_1$ cannot exist, since the corresponding
 Choi matrix is not positive. However, there are channels with all entries
 negative, e.g.~the one with $\bm{\lambda}=(-1,-1,-1)/3$ \cite{Wolf08}.  

There is a connection between the Kraus rank of a channel $\Phi$, defined in
\ref{sec.quch}, and the dimension of the boundary on which its representing
point lies:
\begin{proposition}
For all \emph{unital} qubit channels $\Phi$, the Kraus rank of $\Phi$ is one
plus the dimension of  
the face of the tetrahedron $\mathcal T$ to which the point $\bm{\lambda}=(\lambda_1,
\lambda_2, \lambda_3)$ belongs. Namely, rank-one channels (unitary
conjugations) correspond to the vertices of $\mathcal T$, rank-2 channels
correspond to interior of edges, rank-3 channels correspond to interior of
faces and full-rank channels to the interior of $\mathcal T$. 
\end{proposition}
\begin{proof}
Since positivity of the Choi matrix is equivalent to complete positivity of
the qubit channel, a single vanishing $q_i$ defines a face of the
tetrahedron. Two vanishing 
$q_i$s 
give the intersection of the corresponding two faces, i.e.~an edge, and
three vanishing $q_i$s a vertex. If no $q_i$ is zero, we 
have a generic point inside the tetrahedron. Since at the same time the
number of non-vanishing $q_i$s is the Kraus rank of the channel (the rank of
its Choi matrix) the result holds.
\end{proof}

In the following subsections, we investigate the decomposition of qubit channels with given Kraus rank.

\subsection{Edges of tetrahedron}
\label{sec:edges-tetrahedron}
As mentioned above, all edges of $\mathcal T$ are equivalent up to
permutation of the vertices. Therefore we only need to consider one of the
edges, e.g.~the  
edge $\overline{V_1V_4}$. Points belonging to this edge correspond to
channels $\Phi_{\pf}(t)$ given by $M={\rm diag}(1-2t,1-2t,1)$, $t\in
[0,1]$. These are phase flip channels, where the probability for a phase
flip (conjugation with  
$\sigma_z$) is equal to $t$. This follows from the fact that the vertex
$V_1$ corresponds to the identity channel, whereas the vertex $V_4$
corresponds to a  
unitary conjugation by the $\sigma_z$ Pauli matrix, giving
\begin{equation}\label{phipf}
	\Phi_{\pf}(t) : \rho \mapsto (1-t)\rho + t \sigma_z \rho \sigma_z.
\end{equation}
The entire edge represents the set $\mathcal F_{\pf}=\{\Phi_{\pf}(t), t\in [0,1]\}$. More generally, for $0 < T < 1/2$, we define the restricted set 
\begin{equation}
\mathcal F_{\pf}(T) = \{\Phi_{\pf}(t) \, , \, t \in [0,T]\}
\label{eq:1t}
\end{equation}
Note that the set $\{\Phi_{\pf}(1-t) \, , \, t \in [0, T]\}$ can be obtained from the set  
$\mathcal F_{\pf}(T)$ by unitary conjugation with $U=\exp{({\rm
    i}\,\frac{\pi}{2}\sigma_{z})}$. Since unitary channels are included
in our minimal set $\mathcal F$, it suffices to generate channels
corresponding to the half-axis $T=1/2$. In the case of phase flip channels,
we have the following technical result. 

\begin{lemma}\label{lemtopedge}
For any fixed $0 < T < 1/2$ and a given maximum number $n$ of concatenations, one has 
\begin{equation}
\mathcal F_{\pf}(T)  \subset \langle \mathcal F_{\pf}(\varepsilon)\rangle_n \,.
\end{equation}
with $\varepsilon=\frac{1}{2}(1-(1-2T)^{1/n})$. 
\end{lemma}
\begin{proof}
Concatenation of $n$ phase flips, $\Phi_{\pf}(\varepsilon)^n$, leads to the $M$
matrix $M={\rm diag}((1-2\varepsilon)^n,(1-2\varepsilon)^n,1)$. For any
fixed $T$ with $0<T<1/2$, the first two entries are equal to $1-2T$ when
$\varepsilon=\frac{1}{2}(1-(1-2T)^{1/n})$. Since this is an increasing
function of $T$, for any $t\in[0,T]$ there exists a
$\varepsilon'<\varepsilon$ such that
$\Phi_{\pf}(t)=\Phi_{\pf}(\varepsilon)^n$.   

\end{proof}
If one does not place any limit on the number of concatenations, the situation is much simpler: for any $\varepsilon >0$, we have 
\begin{equation}
	\{\Phi_{\pf}(t) \, , \, t \in [0, 1/2)\} = \langle \mathcal F_{\pf}(\varepsilon) \rangle.
\end{equation}
Together with unitary channels, this set generates all phase flips along the
edge $\overline{V_1V_4}$, apart from the 1/2-phase flip channel. In fact,
$\Phi_{\pf}(1/2)$ must necessarily be included in the universal channel set,
because of the following result.  

\begin{proposition}\label{propFP} 
For any decomposition  $\Phi_{\pf}(1/2)=\Phi_2\circ\Phi_1$, at least one of
$\Phi_1$ or $\Phi_2$ is unitarily equivalent to $\Phi_{\pf}(1/2)$.
\end{proposition}
\begin{proof}
Let $\Phi_{\pf}(1/2)=\Phi_2\circ\Phi_1$ be such a decomposition, and let $M_{1,2}$ be the $3 \times 3$ matrices defining the channels $\Phi_{1,2}$. One has
\begin{equation}\label{eq:PF12}
	M_{{\pf}(1/2)} = \begin{pmatrix}
	0 & 0 & 0 \\
	0 & 0 & 0 \\
	0 & 0 & 1 \end{pmatrix} = M_2 M_1,
\end{equation}
and thus $1=\|M_{{\pf}(1/2)}\| \leq \|M_1\| \|M_2\|$, where $\| \cdot \|$
denotes the usual operator norm $\|M\|=\sup_{x} |Mx|/|x|$. 
Since both norms of $M_{1,2}$ are smaller than or equal to 1, it must be
that $\|M_1\| 
= \|M_2\| = 1$. Without loss of generality, we can assume that
$\lambda_3(M_{1,2}) = 1$. Using the Fujiwara-Algoet conditions, this implies
that the signed singular values of $M_1$ and $M_2$ are of the form
$(a,a,1)$ and $(b,b, 1)$ respectively, for some $a,b$. Taking the
determinant in 
Eq.~\eqref{eq:PF12}, one gets $ab=0$, so that one of $M_1$ or $M_2$ must be
equal to the initial phase flip channel $\Phi_{\pf}(1/2)$, up to unitary
conjugation. 
\end{proof}

As mentioned, all the edges are obtained 
by switching the signs of
eigenvalues, 
permuting the eigenvalues, or by combination of the
two procedures. This way, we only need to include $\mathcal
F_{\pf}(\varepsilon) \cup \{\Phi_{\pf}(1/2)\}$ in our universal set of
quantum channels. We gather the results in this subsection in the following
proposition.  
\begin{proposition}\label{prop:universal-edge}
The unital qubit channels situated on the edges of the tetrahedron $\mathcal
T$ can be obtained by the concatenation of channels from the following
\emph{edge-universal set} ($\varepsilon$ is an arbitrarily small positive constant): 
\begin{equation}
	\mathcal G_{edge}^\varepsilon = \mathcal F_{\pf}(\varepsilon) \cup \{\Phi_{\pf}(1/2)\} \cup \mathcal U_2.
\end{equation}
\end{proposition}

\subsection{A universal set of unital qubit channels}
In the seminal paper \cite{Wolf08}, the divisibility of quantum channels was
investigated. A 
quantum channel $\Phi\in\mathcal {\mathcal C}_d$ was called
\emph{indivisible} if 
every possible decomposition of the form $\Phi=\Phi_2\circ \Phi_1$ is such
that one of the 
$\Phi_i$ is a unitary conjugation. We recall one of the main results from
\cite{Wolf08}. 
\begin{proposition}[\cite{Wolf08}, Theorem 23]\label{propIndv}
A non-unitary qubit quantum channel $\Phi\in
\mathcal C_2$ is indivisible if and only if it has Kraus rank 3.  
\end{proposition}

Therefore, \emph{any} universal set 
of qubit channels
must include the indivisible channels represented by the faces of the
tetrahedron. We shall denote by $\mathcal I_2$ the set of all indivisible
qubit channels 
\begin{equation}
\mathcal I_2 = \{ \Phi \in \mathcal C_2 \, : \, \Phi \text{ has Choi rank 3}\}.
\end{equation}

According to Proposition \ref{propIndv}, all channels on the edges of the
tetrahedron are divisible. However, divisibility does not guarantee
reduction to a more basic set of channels. For instance, the phase flip
channel $\Phi_{\pf}(1/2)$ is divisible in the sense of proposition
\ref{propIndv}, as $\Phi_{\pf}(1/2)=\Phi_{\pf}(1/2)\circ \Phi_{\pf}(1/2)$.  

We now state the main result of this section, an $\varepsilon$-small universal set of unital qubit channels. 

\begin{theorem}\label{thm:universal-unital-qubit}
For any $\varepsilon >0$, the set 
\begin{eqnarray}
	\mathcal G^\varepsilon &=&\mathcal I_2\cup \mathcal G_{edge}^\varepsilon\\&=&\mathcal I_2 \cup \mathcal
	F_{\pf}(\varepsilon) \cup \{\Phi_{\pf}(1/2)\} \cup \mathcal U_2
\end{eqnarray}
is a \emph{universal set} of \emph{unital} qubit channels.
It is minimal in the sense that all elements of $\mathcal I_2$ and
$\{\Phi_{\pf}(1/2)\}$ are needed, as well as a channel
$\Phi_{\pf}(\varepsilon')$ 
where $\varepsilon'\le \varepsilon$.\end{theorem}

\begin{proof}
According to Proposition \ref{prop:universal-edge}, the set $\mathcal
F_{\pf}(\varepsilon) \cup \{\Phi_{\pf}(1/2)\}$ together with unitaries,
generates all channels on the edges of the tetrahedron, i.e.~all Kraus rank
1 and 2 channels.  The set $\mathcal I_2$ contains all Kraus rank 3
channels. It remains to show that the union of these two sets generates all
the Kraus rank 4 channels (interior of the tetrahedron). To this end,
consider the four edges $e_1,\ldots,e_4$ of the tetrahedron with
$e_1=\overline{V_1V_2}$, $e_2=\overline{V_1V_3}$, $e_3=\overline{V_4V_2}$,
and $e_4=\overline{V_4V_3}$, connecting vertices ($V_1,V_4$) to vertices
($V_2,V_3$).  The plane representing the set  
of channels with fixed $\lambda_3=z$, $z\in [-1,1]$, intersects these four
edges in 
four points $A_1,\ldots,A_4$, respectively, with $A_1=(1,z,z)$,
$A_2=(z,1,z)$, $A_3=(-z,-1,z)$, and $A_4=(-1,-z,z)$. These four points form
a rectangle, see Fig.~\ref{fig:tetra}.  

\begin{figure}
\includegraphics{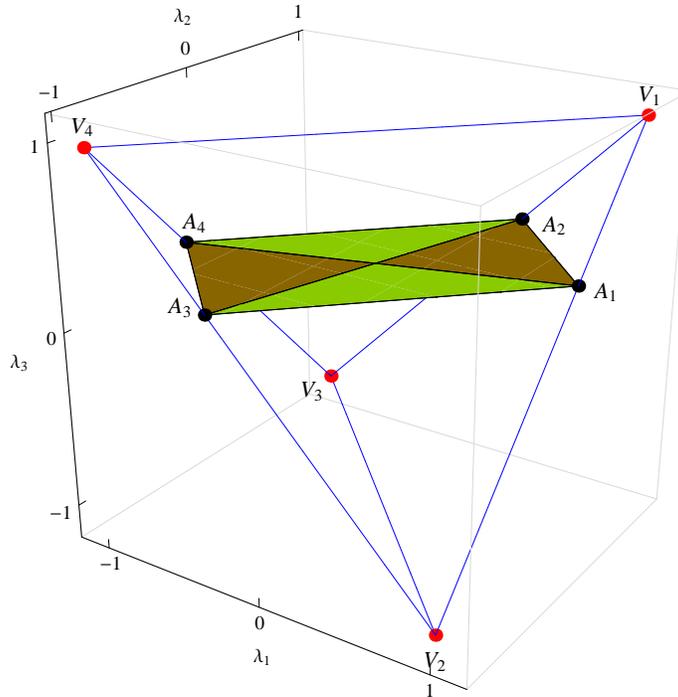}
\caption{The Fujiwara-Algoet tetrahedron $\mathcal T$ of admissible signed
  singular values of a quantum channel with a section $A_1A_2A_4A_3$,
  corresponding to $\lambda_3=1/2$, and the bow-tie regions obtained by
  concatenating channels from the edges. For details see the proof of
  Theorem \ref{thm:universal-unital-qubit}.} 
\label{fig:tetra}
\end{figure}

Consider the set of channels $\mathcal R_1=\{\Phi(s,z),
s\in(0,1)\} \subset \mathcal I_2$, where $\Phi(s,z)$ is the channel associated with the matrix $T$ of the form \eqref{T} with $\bf{t}=\bf{0}$ and $M=M(s,z)={\rm
  diag}(1+s(z-1),z+s(1-z),z)$ defined by the edge $\overline{A_1A_2}$
of the 
rectangle with given $z$. 
  The concatenation
  $\Phi(s,z)\circ\Phi_{\pf}(t)$ 
  with $\Phi(s,z)\in \mathcal R_1$ and 
  $\Phi_{\pf}(t)\in\mathcal F_{V_1V_4}$ (top edge) has the $M$ matrix $M={\rm
    diag}((1-2t)(1+s(z-1)),(1-2t)(z+s(1-z)),z)$. At fixed $s$ the two points
  corresponding to $t=0$ and $t=1$ are on the edge $\overline{A_1A_2}$ and
  $\overline{A_3A_4}$, respectively, and are diametrical with respect to the
  centre of the rectangle at $(0,0,z)$. Because as $t$ varies it linearly
  interpolates between these two points, it fills a line connecting the two
  points. Therefore, when varying 
  $s\in [0,1]$ and $t\in [0,1]$, the channel $\Phi(s,z,t)$ completely fills
  a bow-tie shape (brown/dark-colored region in Fig.~\ref{fig:tetra}),
  corresponding to half of the rectangle $A_1A_2A_3A_4$. The other
  complementary half of the rectangle (green-colored region in
  Fig.~\ref{fig:tetra}) is obtained in a similar manner by concatenating
  channels from the edge $\overline{A_2A_4}$  
described by $M(s,z)={\rm diag}(z+s(-1-z),1+s(-z-1),z)$ with channels
$\Phi_{\pf}(t)$ from $\mathcal F_{V_1V_4}$. Varying $s,z,t$ over their
allowed ranges fills the entire   tetrahedron. 
Since $\Phi(s,z)\in \mathcal I_2$ for $z\in (-1,1)$ and
  $s\in(0,1)$,  and
  since, according to Proposition \ref{prop:universal-edge}, all the
  channels from the top and bottom edges of the 
  tetrahedron can be obtained from $\mathcal
  F_{\pf}(\varepsilon)$ and $\Phi_{\pf}(1/2)$, this completes the proof of
  the universality part of the theorem. 

Regarding the minimality of the set $\mathcal G^\varepsilon$, note that
any universal set needs to contain $\mathcal I_2$ (indivisible channels)
and $\Phi_{\pf}(1/2)$ (because of Proposition \ref{propFP}). All that
remains to be shown is therefore 
that any universal set also needs to contain phase flip channels of
arbitrarily small parameters. To this end, 
consider a non-trivial decomposition of a phase flip channel
$\Phi_{\pf}(\varepsilon) = \Phi_2 \circ 
\Phi_1$. As in Proposition \ref{propFP}, since  the matrix $M$ associated to
the phase flip channel $\Phi_{\pf}(\varepsilon)$ has 
operator norm 1, both $\Phi_1$ and $\Phi_2$ need to be, up to unitary
conjugations, phase flip channels also, of respective parameters $\delta_1$,
$\delta_2$. Taking the determinant in the equation $M = M_2 M_1$, we get
$(1-2\varepsilon)^2= (1-2\delta_1)^2(1-2\delta_2)^2$, so that 
at least one of $\delta_{1,2}$ has to be smaller than
$\varepsilon$. Therefore, any universal set needs to contain
$\Phi_{\pf}(\delta)$ with $\delta < \epsilon$, finishing the proof for the
optimality of the set $\mathcal G^\varepsilon$.
\end{proof}

\section{Geometry of non-unital qubit channels}

\label{nonunital}
We now consider the more general case of non-unital channels. We first
provide two relatively simple forms of generalized Fujiwara-Algoet
conditions which allow one to determine the combinations of $\bt$ and $M$ that
represent completely positive maps.  We use these conditions to subsequently
classify qubit channels by their pure output (PO). Such classification is useful
because concatenation of channels results in a channel whose number of pure
state outputs can be at most equal to the minimal number of pure state
outputs among the used channels. 
\subsection{Condition for complete positivity of non-unital channels}\label{sec:generalized-FAC}
The Choi matrix \eqref{choi} for a general qubit channel with a $T$ matrix
of the form  
(\ref{T}) is given by \cite{Bengtsson06}  
\begin{equation} \label{ChoiT}
C_\Phi=\left(
\begin{array}{cccc}
\frac{1}{2}(1+\lambda_3+t_3)&0&\frac{1}{2}(t_1+it_2)&\frac{\lambda_1+\lambda_2}{2}
\\
0&\frac{1}{2}(1-\lambda_3+t_3)&\frac{\lambda_1-\lambda_2}{2}&\frac{1}{2}(t_1+it_2)
\\
\frac{1}{2}(t_1-it_2)&\frac{\lambda_1-\lambda_2}{2}&\frac{1}{2}(1-\lambda_3-t_3)&0\\
 \frac{\lambda_1+\lambda_2}{2}&\frac{1}{2}(t_1-it_2)&0&\frac{1}{2}(1+\lambda_3-t_3) 
\end{array}
\right)\,.
\end{equation}
By a simple change of basis $RC_\Phi R^{\dagger}$, with 
\begin{equation} \label{basischange}
R=\frac{1}{\sqrt{2}}\left(
\begin{array}{cccc}
1&0&0&1\\
0&1&1&0\\
0&-i&i&0\\
1&0&0&-1
\end{array}
\right)\, ,
\end{equation}
the Choi matrix can be rewritten as
\begin{equation} \label{ChoiTbis}
C_\Phi=\frac12\left(
\begin{array}{cccc}
4q_0&t_1&t_2&t_3\\
t_1&4q_1&i t_3&-i t_2\\
t_2&-i t_3&4q_2&i t_1\\
t_3&i t_2&-i t_1&4q_3
\end{array}
\right)\, ,
\end{equation}
where $q_i$ are the linear combinations of $\lambda_i$ introduced in \eqref{qi_0}-\eqref{qi_3}. 
Note that a necessary condition for $C_\Phi\ge 0$ is $q_i\ge 0$ for all
$i=0,\ldots,3$, so that $\lambda_i$ still satisfy the original FAC
conditions \eqref{FAC}.  The following result generalizes the FAC to the case of non-unital channels, beyond the simple case where only one of the $t_i$ is non-zero. Equivalent necessary and sufficient conditions were obtain in 
\cite[Corollary 2]{Ruskai02}, in the form of three inequalities; we claim that our re-formulation has a more natural geometric interpretation. 
\begin{theorem}[Generalized Fujiwara-Algoet conditions]\label{thm:GFA}
Let $\Phi: \mathcal M_2(\mathbb C) \to \mathcal M_2(\mathbb C)$ be a non-unital
linear map whose matrix in the Pauli basis is given by
\eqref{eq:T-diagonal}. Let $t=\|\bt\|$ and 
$\bu=\bt/t$ the corresponding unit vector. Then the map
$\Phi$ is a quantum channel 
if and only if 
\begin{align}
q_i &\geq 0, \quad i=0,1,2,3 \qquad \text{ and}\nonumber \\
t^2 &\leq r-\sqrt{r^2-q},
\label{descartes5}
\end{align}
where  the $q_i$ are defined in
\eqref{qi_0}-\eqref{qi_3} and 
\begin{align} 
\label{eq:def-r} r&=1-\sum_i\lambda_i^2+2\sum_i\lambda_i^2u_i^2,\\
\label{defp} q&=256\prod_{i=0}^3q_i.
\end{align} 
\end{theorem}
\begin{proof}
Since trace preservation  of the map $\Phi$ follows from
\eqref{eq:T-diagonal}, the only property that needs to be checked is
complete positivity. By Choi's theorem, $\Phi$ is completely positive if and
only if the Choi matrix $C_\Phi$ is positive. The characteristic polynomial
$p(x)=\det(C_\Phi-xI_4)$ of $C_\Phi$ reads 
\begin{equation} \label{px}
p(x)=x^4-2x^3+\frac{a}{2}x^2-\frac{b}{2} x+\det C_\Phi\,,
\end{equation}
where 
\begin{eqnarray}
a&=&3-\sum_i\lambda_i^2-t^2\label{coeffa}\\
b&=&1-\sum_i\lambda_i^2-t^2+2\lambda_1\lambda_2\lambda_3\label{coeffb}\\
\det C_\Phi&=&\frac{1}{16}(t^4-2rt^2+q\label{detC})\,.
\end{eqnarray}
Since $C_\Phi$ is Hermitian its roots are real. By Descartes' rule of signs, all roots $x_i$ are positive iff the coefficients
of the powers of $x$ change sign from one coefficient to the next, that is,
iff $\det C_\Phi\ge 0$, $a\ge 0$, and $b\ge 0$. Since $q_i$ are diagonal elements of $C_\Phi$ in \eqref{ChoiTbis},
a necessary condition for positivity of $C_\Phi$ is that $q_i$ be all positive,
that is, $\lambda_i$ lie within the tetrahedron $\mathcal T$ of admissible
values of the 
unital case. As $q_i\ge 0$ ($i=0,\ldots,3$)
implies $|\lambda_i|\le 1$ ($i=1,\ldots,3$) and thus
$|\prod_{i=1}^3\lambda_i|\le 1$, one always has $a\geq b$, thus condition
$a\ge 0$ is a consequence of $b\ge 0$. We are therefore left with just two
generalized Fujiwara-Algoet (GFA) conditions,  
\begin{eqnarray}
\det C_\Phi&\ge &0\,\,\,
\Leftrightarrow \,\,\, t^4-2rt^2+q\geq 0, \mbox{ \quad and}\label{descartes3}\\
b&\ge &0\,\,\,
\Leftrightarrow \,\,\, t^2\leq
1-\sum_i\lambda_i^2+2\lambda_1\lambda_2\lambda_3\,.\label{descartes2} 
\end{eqnarray}

Note that these conditions only depend on the square of the $t_i$, whereas
the signs of the $\lambda_i$ matter. 
Condition \eqref{descartes3} is a polynomial of degree 2 in $t^2$, whose
discriminant $r^2-q$ is always positive. Indeed, let $\lambda_{3}$ be the
signed singular value with the smallest absolute value; then $r$, as a
function of $\mathbf{u}$, reaches its minimal value which is
$r_{\textrm{min}}=1-\lambda_1^2-\lambda_2^2+\lambda_3^2$. Therefore,
\begin{equation}
\label{qdeux}
r^2-q\geq r_{\textrm{min}}^2-q=4(\lambda_1\lambda_2-\lambda_3)^2\geq 0.
\end{equation}
Because of positive discriminant the polynomial in (\ref{descartes3}) has two real roots $r \pm \sqrt{r^2-q}$ and, due to positive coefficient in front of $t^4$, the condition \eqref{descartes3} is thus satisfied iff
\begin{equation}
t^2\leq r-\sqrt{r^2-q}\mbox{ or\quad} t^2\geq r+\sqrt{r^2-q}.
\label{eq:roots}
\end{equation}
We shall now show that the condition $b \ge 0$ selects the left root as the relevant one.

It turns out that, fixing $\lambda_j$, for any value of the $u_i$ two further inequalities hold,
\begin{equation}
\label{ineq}
 r-\sqrt{r^2-q}\leq  1-\sum_i\lambda_i^2+2\lambda_1\lambda_2\lambda_3 \leq
 r+\sqrt{r^2-q}\ . 
\end{equation}
To show these, one first notices that the quantities $r-\sqrt{r^2-q}$ and
$r+\sqrt{r^2-q}$ are, respectively, decreasing and increasing functions of
$r$ ($r-\sqrt{r^2-q}$ is decreasing because, taking a derivative, we get the
condition $\sqrt{r^2-q}\le r$, which is always satisfied within the
tetrahedron). Therefore, $r-\sqrt{r^2-q}$ is always smaller than or equal to
$r_\textrm{min}-\sqrt{r_\textrm{min}^2-q}$, while $r+\sqrt{r^2-q}$ is always
larger than or equal to  $r_\textrm{min}+\sqrt{r_\textrm{min}^2-q}$, where
$r_{\textrm{min}}=1-\lambda_1^2-\lambda_2^2+\lambda_3^2$ is the smallest
possible value of $r$ when $u_i$ varies (we again denote by $\lambda_3$ the
one with the smallest absolute value). Inequalities \eqref{ineq} will
therefore follow if we show that 
\begin{equation}
\label{ineqmin}
 r_{\textrm{min}}-\sqrt{r_{\textrm{min}}^2-q}\leq  1-\sum_i\lambda_i^2+2\lambda_1\lambda_2\lambda_3 \leq r_{\textrm{min}}+\sqrt{r_{\textrm{min}}^2-q}.
\end{equation}
Showing (\ref{ineqmin}) is equivalent to showing $f_- \ge 0$ and $f_+\le 0$, where we defined $f_\pm=1-\sum_i\lambda_i^2+2\lambda_1\lambda_2\lambda_3-( r_{\textrm{min}}\pm\sqrt{r_{\textrm{min}}^2-q})$. Plugging explicit expressions for $q=(1+\lambda_1+\lambda_2+\lambda_3)(1+\lambda_1-\lambda_2-\lambda_3)(1-\lambda_1+\lambda_2-\lambda_3)(1-\lambda_1-\lambda_2+\lambda_3)$ and $r_\textrm{min}=1-\lambda_1^2-\lambda_2^2+\lambda_3^2$ into $f_\pm$, and simplifying, results in
\begin{eqnarray}
f_-
&=& 2|\lambda_1\lambda_2-\lambda_3|+2\lambda_3(\lambda_1\lambda_2-\lambda_3),\nonumber \\
f_+&=& -2|\lambda_1\lambda_2-\lambda_3|+2\lambda_3(\lambda_1\lambda_2-\lambda_3).
\end{eqnarray}
As $|\lambda_3|\leq 1$, indeed $f_- \ge 0$ and $f_+ \le
0$. Inequalities (\ref{ineqmin}) are therefore true, and so are
(\ref{ineq}). 
The logic of the two directions of the proof of Theorem \ref{thm:GFA}
can now be seen summarized as 
follows:\\
1.) If $\Phi$ is a quantum channel, then $C_\Phi\ge
0$ and thus inequalities (\ref{descartes3}) and
(\ref{descartes2}) hold. In addition, since $q_i$ are the diagonal matrix
elements of $C_\Phi$ in an appropriate basis (eq.(\ref{ChoiTbis})), we
also have $q_i\ge 0$ and thus $q\ge 0$ which implies inequalities
(\ref{ineq}). Since by (\ref{descartes2}) $t^2\le 1-\sum_i
\lambda_i^2+2\lambda_1\lambda_2\lambda_3$, we have from
(\ref{ineq}) that $t^2\le 1-\sum_i
\lambda_i^2+2\lambda_1\lambda_2\lambda_3\le r+\sqrt{r^2-q}$. Therefore, the
second inequality in (\ref{eq:roots}) is satisfied only when
$t^2=r+\sqrt{r^2-q}$. This equality implies equality in the second
inequality of
 (\ref{ineq}), 
which in turn implies that $r=r_{\textrm{min}}$. Therefore,
$f_{+}=0$ and thus $|\lambda_i|=1$ which corresponds to a unitary, and thus unital,
channel, which is excluded in the statement of the Theorem. The only remaining possibility is that the left inequality in
(\ref{eq:roots}) be satisfied. In any case we 
have that $C_\Phi\ge 0$ which implies  $t^2\leq r-\sqrt{r^2-q} \mbox{ and }
q_i\ge 0$. \\
2.) If $t^2\leq r-\sqrt{r^2-q}$ we have from
  (\ref{eq:roots}) that $\det C_\Phi\ge 0$. Furthermore, since by
  assumption $q_i\ge 
  0$, inequalities (\ref{ineq}) are valid and we have thus $t^2\le
  r-\sqrt{r^2-q}\le 1-\sum_i
\lambda_i^2+2\lambda_1\lambda_2\lambda_3$. The latter chain of inequalities
  implies by  
  (\ref{descartes2}) that $b\ge 0$, which together with $\det
  C_\Phi\ge 0$ gives $C_\Phi\ge 0$ and thus the complete positivity of
  channel $\Phi$. 
\end{proof}

Let us make now some remarks on the conditions appearing in the theorem above. First, note that we exclude unital channels ($t=0$), since in that case the vector $\mathbf{u}$ is not defined; one can use the usual Fujiwara-Algoet conditions \eqref{FAC} in that case. Also, note that for any fixed set of $\lambda_j$, the condition $\det C_\Phi \ge 0$ is
necessary, but not sufficient. The set of translation vectors
$\bt$ satisfying $\det C_\Phi \ge 0$ is composed of a bounded
part corresponding to $t^2\leq r-\sqrt{r^2-q}$, and an unbounded part
with $t^2\geq r+\sqrt{r^2-q}$. The second condition $b \ge 0$ then
selects $t^2\leq r-\sqrt{r^2-q}$ as the one resulting in a completely
positive map. 

Conditions equivalent to \eqref{descartes5} were already found in \cite[Corollary 2]{Ruskai02}, in the form of three inequalities. Inequality \eqref{descartes5} gives the
maximal modulus square of the translation vector ${\bf t}$ of the ellipsoid
compatible with positivity of $C_\Phi$. This has a very natural geometric interpretation: it gives the maximum displacement of a given ellipsoid in a given direction such that the corresponding linear map is completely positive.
In particular, if one of the $q_i$
is zero (i.e., the corresponding channel is represented by a point
on a face of the tetrahedron), then $q=0$ and the condition
\eqref{descartes5} implies that one must have $\bt=0$. The
implication holds in the opposite direction, namely, if the GFA
conditions \eqref{descartes5} only allow $\bt=0$, then $q=0$. Thus, one of
the $q_i$ vanishes, implying that $\bm{\lambda}$ is a point on the surface
of the 
tetrahedron. As soon as $\bm{\lambda}$ is 
within the tetrahedron the right-hand side of the GFA condition 
(\ref{descartes5}) is nonzero and also non-unital channels with such
$\lambda$s exist.

Note also that in the particular case discussed in \cite{Fujiwara99}, where $t_1=t_2=0$ (so that, in our notation, $u_1=u_2=0$ and $u_3=1$), the quantity $\sqrt{r^2-q}$ appearing in the second equation in \eqref{descartes5} simplifies to 
\begin{equation}
\sqrt{r^2-q}= 2|\lambda_3 - \lambda_1\lambda_2|,
\end{equation}
in such a way that \eqref{descartes5} is equivalent to the (generalized) FAC
\begin{equation}
t_3^2 \leq (1 \pm \lambda_3)^2 - (\lambda_1 \pm \lambda_2)^2.
\end{equation}

\subsection{Classifying channels by their pure output}

The goal of this subsection is to classify all qubit channels by the number of pure outputs that they can have. Our main result, Theorem~\ref{thm:pure-output-alternative}, prohibits qubit channels (apart from unitary conjugations) that would have more than 2 pure outputs. We prove the theorem by elementary geometric means, however, note that it follows also from the results presented in ~\cite{Ruskai02}. Namely, it has been shown~\cite{Ruskai02} that extreme or quasi-extreme channels (i.e., those from the interior of tetrahedron edges) that are not unitary conjugations can have at most 2 pure outputs. Because every channel can be written as a convex combination of extremal channels, as soon as one of the channels in the convex sum is not a unitary conjugation, at most 2 pure outputs are possible. If we have a convex combination of unitary conjugations only, then we know that a convex combination of two unitary conjugations is a quasi-extremal channel, again having at most 2 
 pure outputs, leading to the same conclusion.

\begin{definition}
The \emph{pure output} (PO) of a quantum channel $\Phi$ is the set of
pure states 
in the image of $\Phi$: 
\begin{equation}
	\mathrm{PO}(\Phi) = \Phi(\mathcal D_d) \cap \mathcal P_d\,.
\end{equation}
\end{definition} 
Recall that $\mathcal D_d$ is the set of density matrices and $\mathcal P_d\subset
\mathcal D_d$ the set of pure states. 

The pure output of a unital channel inherits the central symmetry of the
output ellipsoid in the following precise sense:

\begin{lemma}\label{lem:PO-symmetry}
The pure output of a unital channel $\Phi$ is centrally symmetric,
i.e.
\begin{equation}
\rho(\br)=\frac{1}{2}(I_2+\br.\bm{\sigma}) \in \mathrm{PO}(\Phi)
\Leftrightarrow \rho(-\br)=\frac{1}{2}(I_2-\br.\bm{\sigma}) \in
\mathrm{PO}(\Phi).  
\end{equation}
\end{lemma}
\begin{proof}
Since both $\Phi(\mathcal D_d)$ and $\mathcal P_d$ are centrally symmetric, so is their intersection.
\end{proof}

In the following we show that quantum channels can be classified according
to their pure output. An arbitrary single qubit channel maps the input
states -- a Bloch ball -- to output states within a shifted ellipsoid
\cite{King01,Ruskai02}. Therefore, we start by proving the following
elementary Euclidean geometry results that will help us understand possible
intersections between a sphere and an ellipsoid. In what follows, we shall
abuse language and say that a set is \emph{contained inside} a circle
(resp. a sphere) if it is a subset of the corresponding disc (resp. ball).
\begin{lemma}
\label{lemcircle}
Let $C$ be a circle in $\mathbb R^2$ and consider
 an 
ellipse $E$ contained inside the circle. If the intersection 
$C \cap E$ contains
three distinct points $M,N,O$, then $E=C$.  
\end{lemma}
\begin{proof}
This immediately follows from the fact that an ellipse (and, more generally,
a conic section) is uniquely determined by the condition that it passes
through three non-collinear points and is tangent to two given lines passing
through two of these points (for a proof, see
\cite{Olivier1845}, p. 114). Since the hypothesis
implies that both $E$ and $C$ have this property, one must have $E=C$ by
uniqueness. 
\end{proof}

\begin{lemma}\label{lem3P}
Let $S$ be the unit sphere in $\mathbb R^3$ and consider
an ellipsoid $E$ contained inside the sphere. If the intersection $S \cap E$
contains 
three distinct points $M,N,O$, then it contains the circle determined by
those three points. 
\end{lemma}
\begin{proof}
Consider the plane determined by the points $M,N,O$. Its intersection with
the sphere $S$ defines a circle $C$ and its intersection with the ellipsoid
$E$ defines an ellipse $F$. Obviously, $M,N,O \in C \cap F$ and $F$
is inside in 
$C$. Using lemma \ref{lemcircle}, we have $C=F$, which is the conclusion.  
\end{proof}

\begin{proposition}\label{ES}
Let $S$ be the unit sphere in $\mathbb R^3$ and consider  an
ellipsoid $E$ contained inside $S$. Then, the intersection $S \cap E$ 
is one of the following:  
\begin{enumerate}
\item the empty set;
\item a point;
\item two points;
\item a circle;

\item the whole sphere $S$.
\end{enumerate}
\end{proposition}
\begin{proof}
Using the previous lemma, we only need to consider the case when the
intersection contains 4 non-coplanar points $M,N,O,P$ and to show that the ellipsoid coincides with the sphere. By the results
already proved, the intersection contains in fact the whole circle $C$
determined by $M,N,O$. Let $Q$ be any point on the sphere; we will show that
$Q$ belongs to the intersection $S \cap E$. To this end, consider the plane
determined by the points $P, Q$ and $F$, where $F$ is the center of the
circle $C$. This plane cuts the sphere in a circle $C'$, which intersects
$C$ in two points $X$ and $Y$. Since $X, Y$ and $P$ belong to the
intersection $S \cap E$, so does the circle $C'$ (by Lemma \ref{lem3P}), and
thus, in particular, the point $Q$, finishing the proof. 
\end{proof} 

The case of an intersection in the form of a circle can be further
restricted. First, we state the following lemma, which is Theorem LXXIII in
the Supplement of \cite{Allen1822}. 

\begin{lemma}\label{euclid}
Let $C$ be a circle and $E$ an ellipse contained inside the circle and
touching the circle from the inside in exactly two points $P,Q$. Then the large
ellipse axis is parallel to $PQ$, while the line of the small axis bisects the segment $PQ$. 
\end{lemma}

A circle intersection is now of the following type.
\begin{lemma}\label{symell}
The only way for an ellipsoid $E$ with half axes $0<a,b,c\le 1$ to touch the
sphere $S$ from the inside in the form of a circle of non-zero radius, is to have an ellipsoid that is rotationally symmetric about one of 
its three axes and displaced in the direction of this axis.
\end{lemma}

\begin{proof}
Let us call $P$ the plane containing the touching circle and $C$ the circle
  centre. We orient the $z$-axis of $\mathbb R^3$ so that it is
  perpendicular to $P$ and  
  passes through $C$, while the origin $O$ of $\mathbb R^3$ is the center of the sphere $S$. First, we shall show that 
one of the ellipsoid's axes has to be
the $z$-axis. Then we will show that the ellipsoid has to be rotationally
symmetric about the $z$-axis.

Let us consider an arbitrary plane $R$ containing the $z$-axis. Such a plane
cuts the sphere $S$ in a 
  circle, while it cuts the ellipsoid $E$ in an ellipse. It also contains
  two points $M$ and $N$ from the intersecting circle $C$. Therefore, in 
  plane $R$ one has an ellipse that touches a circle from inside in exactly
  two points, $M$ and $N$. Note that the $z$-axis bisects the segment
  $MN$ perpendicular to it. From Lemma~\ref{euclid}, it follows that the small ellipse 
  axis is the $z$-axis. We shall consider now two such particular planes $R$. 
  
First, choose $R_1$ to be the plane containing the $z$ axis and the
center $L$ of the ellipsoid. It is a known fact that the intersection
of the ellipsoid $E$ with the plane passing through its center $L$ is
an ellipse with center $L$. Since the small axis of the ellipse is $Oz$,
we infer $L \in Oz$, so the center of the ellipsoid lies on the $z$
axis.    

Let us now choose a second plane $R_2$, containing the $z$ axis
(and thus the center $L$ of $E$) and the smallest axis of
the ellipsoid $E$. The intersection $E \cap R_2$ contains the points
$X,Y$ which are the antipodal points ($L$ is the middle of $XY$) closest
to each other of $E$. They are also the points closest to each other and
symmetric with respect to $L$ of the 
ellipse  $E \cap R_2$. It follows that $X,Y \in
Oz$ so that one of the axes of the ellipsoid is the $z$ axis.  

Points on the surface of the ellipsoid $E$
  therefore satisfy 
\begin{equation}
\frac{x^2}{a^2}+\frac{y^2}{b^2}+\frac{(z-z_0)^2}{c^2}=1,
\end{equation}
while the touching circle can be parametrized as
$(r\cos{\phi},r\sin{\phi},z_1)$ with some nonzero $r$ and fixed $z_1$. For
any $\phi$ these points should lie on the surface of the ellipsoid,
therefore 
\begin{equation}
r^2+(a^2/b^2-1)r^2\sin^2{\phi} =a^2-\frac{a^2}{c^2}(z_1-z_0)^2
\end{equation}
should hold. The RHS is independent of $\phi$ and so should be the
LHS. Therefore, we conclude that $a=b$. The ellipsoid must be rotationally
symmetric, and its displacement can be only along the symmetry axis $Oz$. 
\end{proof}

One might think that qubit channels can realize the five different
types of pure output suggested by Lemma \ref{ES}.  However, we will now show
that this is not the case. Rather, the following result holds (see Figure
\ref{fig:PO}): 

\begin{figure}
\includegraphics[width=0.4\textwidth]{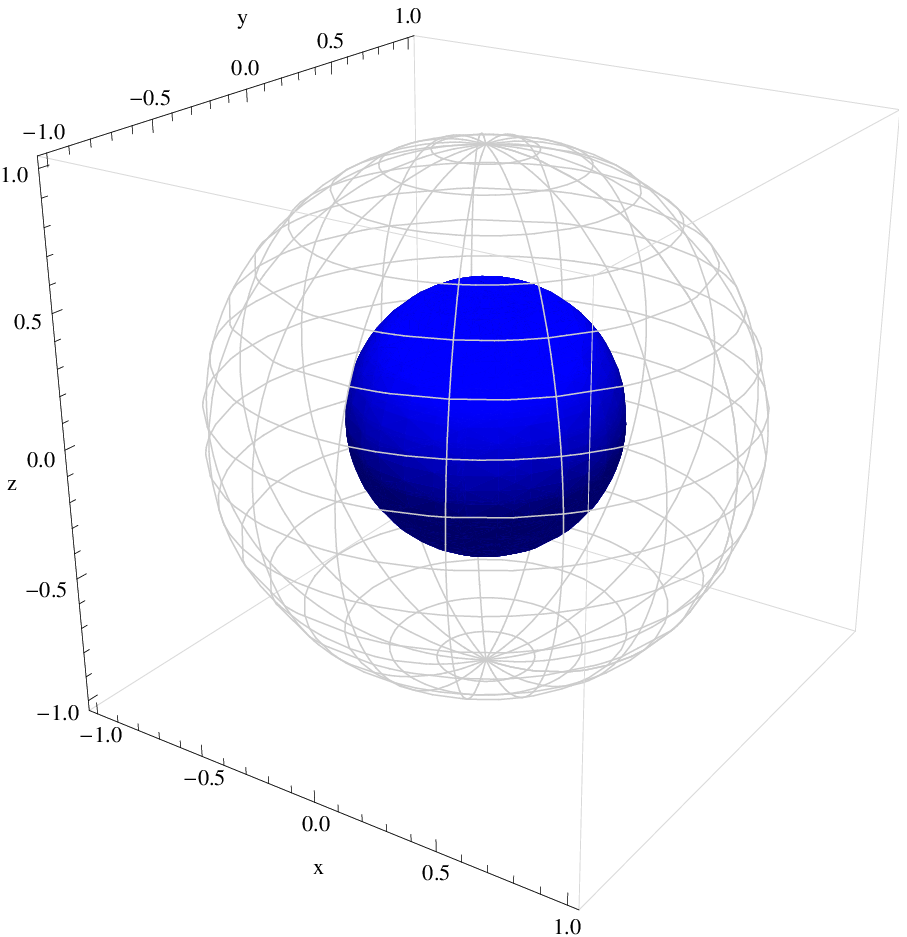}
\includegraphics[width=0.4\textwidth]{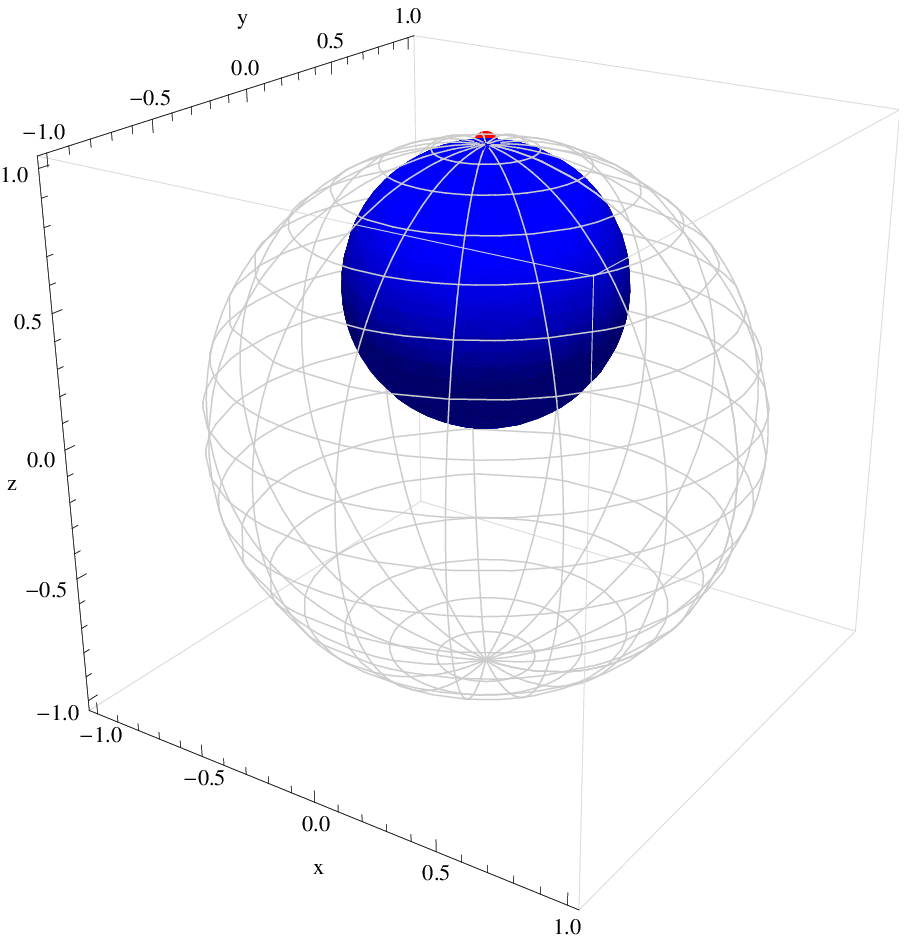}\\
\includegraphics[width=0.4\textwidth]{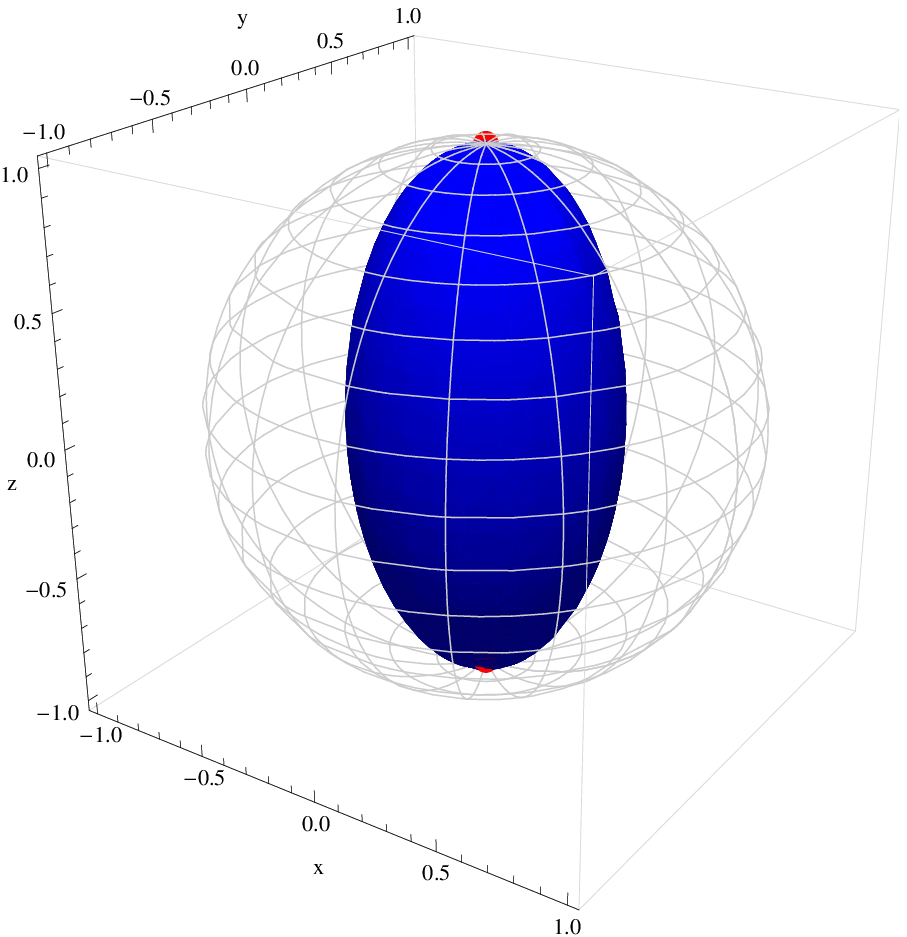}
\includegraphics[width=0.4\textwidth]{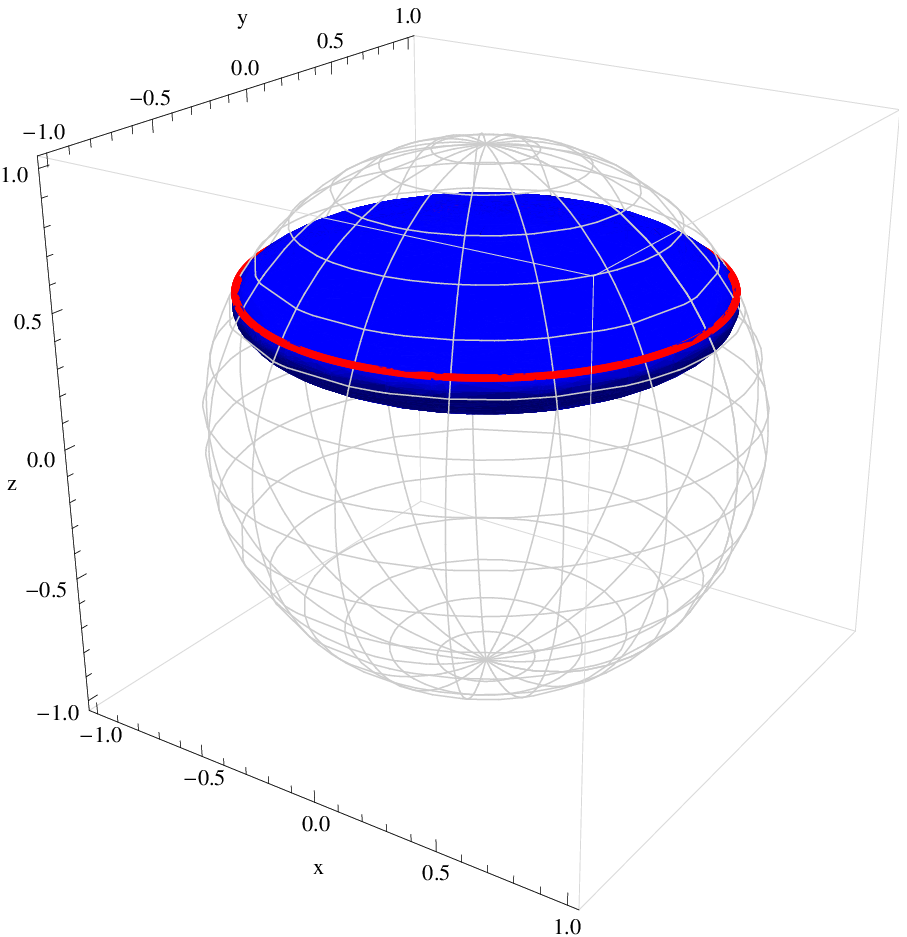}
\caption{Different ellipsoids inside the Bloch sphere with,
  respectively, empty, 1 point, 2 points and circular pure outputs that
  could be possible from purely geometrical considerations.
The
  circular case does not correspond to the output of a completely positive
  quantum channel.} 
\label{fig:PO}
\end{figure}

\begin{theorem}\label{thm:pure-output-alternative}
Let $\Phi \in \mathcal C_2$ be a qubit channel. One of the following holds: 
\begin{enumerate}
	\item $\mathrm{PO}(\Phi) = \emptyset$, the channel has no pure
	output, all output states are mixed; 
	\item $\mathrm{PO}(\Phi) = \{\xi\}$, $\xi \in \mathcal P_2$, the
	channel has a unique pure output $\xi$; 
	\item $\mathrm{PO}(\Phi) = \{\xi,\zeta\}$, $\xi,\zeta \in \mathcal
	P_2$, the 
	channel has exactly two pure outputs $\xi,\zeta$; 
	\item $\mathrm{PO}(\Phi) = \mathcal P_2$, all pure states are
	outputs of $\Phi$. In this case, $\Phi$ is a unitary conjugation
	$\Phi(X) = UXU^\dagger$, for some unitary matrix $U$.  
\end{enumerate}
\end{theorem}
\begin{proof}
The only allowed forms of pure outputs have the geometric forms given in
 Proposition \ref{ES}. 
 Examples of pure outputs different from a circle are easily found
 \cite{Bengtsson06}: The 
 fully mixing channel $\rho\to {I}_2/2$
 maps the entire Bloch sphere to
 its center so that the pure output is the empty set. A 
decaying channel leads to an ellipsoid that touches the sphere 
 in the South pole and nowhere else.  A phase flip channel
 $\Phi_{PF}$ (see eq.(\ref{phipf})) shrinks the Bloch sphere 
 in $x$ and $y$-directions, but leaves the $z$-direction untouched, such
 that the resulting ellipsoid touches the Bloch sphere in the North and
 South poles and nowhere else.  Finally, unitary conjugation corresponds to
 a rotation of the Bloch sphere, and thus has as pure output all pure
 states $\mathcal P_2$. It remains to be shown that a pure output in the
 form of a circle on the Bloch
 sphere does not correspond to a completely positive qubit quantum channel.
According to Lemma \ref{symell}, the demonstration can be reduced to
quantum channels with  
$t_{1}=t_{2}=0$ and $|\lambda_1|=|\lambda_2|=a$, $|\lambda_3|=c$. Note that
 $a,b,c$ are geometrical quantities (half-axes of the ellipsoid) and are
 always positive, whereas the signed singular values $\lambda_i$ can have
 either sign.  In order for the
 ellipsoid to touch the sphere in a circle (and thus at $x^2+y^2>0$),
 one needs 
 $c<a$.  Due to 
 rotational symmetry we need only consider the plane $y=0$ in order to
 obtain a relation between $t_3,a,c$ required for the ellipsoid to touch the
 sphere from inside. 
The $z$ coordinates of the ellipsoid and sphere in the upper half space
read, respectively,
\begin{eqnarray}
z_{E} & = & c\sqrt{1-\left(\frac{x-t_1}{a}\right)^{2}-\left(\frac{y-t_2}{b}\right)^{2}}+t_3\\
z_S & = & \sqrt{1-x^{2}-y^{2}}\,.\label{4}
\end{eqnarray}
If $E$ touches $S$ in a point $x,y,z$, we must have $z_E=z_S=z$ in that
 point, and the tangential planes to $S$ and $E$ in that point must be
 identical.
Two non-trivial solutions of  
 the touching condition $dz_E/dx=dz_S/dx$ are then found
 as $x=\pm\sqrt{\frac{c^2-a^4}{c^2-a^2}}$. In order that the solutions be
 real, we need $c<a^2$. Reinserting this into the second touching condition
 $z_E=z_{S}$ leads to the shift  
\begin{equation} \label{t3}
t_3=\frac{\sqrt{(1-a^2)(a^2-c^2)}}{a}\,.
\end{equation}
The GFA condition \eqref{descartes5} in the present case reads $t^2\leq
(c\pm1)^2-4a^2$, where $\pm$ comes from the two possible signs of
$\lambda_3$. Since here $t^2=t_3^2$ and  
\begin{equation}
(c \pm 1)^2-4a^2-t_3^2=\frac{c^2\pm 2c a^2-3a^4}{a^2}\leq 0
\end{equation}
for any $c<a^2$, the GFA condition can never be satisfied, and thus the 
solutions are not qubit channels. Therefore there is no quantum channel with
a pure output in the form of a circle on the Bloch sphere. 
\end{proof}

Note that when comparing Proposition \ref{ES} with Theorem
\ref{thm:pure-output-alternative}, one sees that the case of the
circle is not physical, i.e. there is no completely positive map which
has a circle as the pure output set. This is a generalization of the ``no
pancake'' theorem, which states that there is no qubit quantum channel
that maps the Bloch sphere to a disk touching the sphere (see \cite{Blume-Kohout10}).

Channels of type 3 in Theorem \ref{thm:pure-output-alternative} can be
either unital or non-unital. For distinguishing the two the following
proposition is useful. 
\begin{proposition}\label{prop:orthog}
If the two output states $\xi,\zeta$ in type 3 states from Theorem
\ref{thm:pure-output-alternative} are orthogonal (i.e.~antipodal on the
Bloch sphere), then the channel is unital, otherwise it is non-unital.  
\end{proposition}
\begin{proof}
The two most distant points on an ellipsoid are on its largest axis, so the
center of the ellipsoid is the middle of the segment $\xi\zeta$, i.e. the
center of the Bloch sphere, if $\xi,\zeta$ are antipodal. It follows that
the channel must be unital, $t = 0$.

The other direction follows from Lemma \ref{lem:PO-symmetry}: the pure output is non symmetric, so the channel cannot be unital.
\end{proof}

Lemma \ref{lem:PO-symmetry} can be used to show that all channels with one PO,
i.e. of type 2 in Theorem \ref{thm:pure-output-alternative}, are
non-unital. Channels with zero PO can be either unital or non-unital.  

\section{Universal set for extremal qubit channels}

We investigate in this section the important role \emph{extremal} 
quantum qubit channels have to play with respect to divisibility and 
universal families.

For general dimensions, necessary and sufficient conditions for a 
quantum channel to be extreme have been found by \cite{Landau-Streater93}, using ideas from \cite{Choi75}.

\begin{theorem}[\cite{Landau-Streater93}]
A quantum channel $\Phi$ having Kraus operators $\{A_i\}_{i=1}^k$ is an 
extremal point of the convex set of quantum channels iff the set of 
matrices $\{A_i^\dagger A_j\}_{i,j=1}^k$ is linearly independent.
\end{theorem}

This result was used in Ref.~\cite{Ruskai02} to provide a more 
geometric picture in the qubit case, which we recall below.
\begin{proposition}[\cite{Ruskai02}]\label{prop:extremal-channels}
A map $\Phi :\M_2(\C) \to \M_2(\C)$ as in \eqref{eq:T-diagonal} is an 
extremal quantum channel if, up to some permutation of indices,
\begin{align}
\lambda_3  &= \lambda_1 \lambda_2, \\
t_3^2 &= (1-\lambda_1^2)(1-\lambda_2^2), \quad t_1=t_2=0.
\end{align}
This is equivalent to the existence of angles $u \in [0, 2\pi)$ and $v 
\in [0, \pi)$ such that
\begin{align}
\lambda_1 &= \cos u,\nonumber \\
\lambda_2 &= \cos v,\nonumber \\
\lambda_3 &= \cos u \cos v,\nonumber \\
t_3 &= \sin u \sin v.
\label{eq:extremal}
\end{align}
Without sacrificing generality one can assume $|\lambda_1| \ge |\lambda_2|$ in the above parametrization. Channels with $u=0$ and $v \neq 0$ are the phase flip channels considered in Section \ref{sec:edges-tetrahedron}; they are not extremal channels and must be excluded from the set of parameters above.
\end{proposition}
For generic angles such a channel maps the Bloch 
sphere to a shifted ellipsoid that touches the sphere in two points (the two points might degenerate to one point for some special values of the angles). As shown 
in Ref.~\cite{Ruskai02}, two pure input states at points 
$(\pm \cos{\theta},0,\sin{\theta})$ are mapped to two pure output states at 
points $(\pm \cos{\omega},0,\sin{\omega})$. Angles $\theta$ and $\omega$ can 
be interpreted as latitudes of the pure inputs and outputs. They are
related to $u$ and $v$ by $\sin{\theta}=\tan{u}/\tan{v}$ and  
$\sin{\omega}=\sin{u}/\sin{v}$, or, inversely, channel matrix elements are 
$\cos{u}=\cos{\omega}/\cos{\theta}$, $\cos{v}=\tan{\theta}/\tan{\omega}$ 
and $t_3=(\cos^2{\theta}-\cos^2{\omega})/(\cos^2{\theta}\sin{\omega})$. 
Sometimes it will be more useful to use a parametrization with angles 
$\theta$ and $\omega$ instead of $u,v$.

Regarding the signs of generalized singular values $\lambda_j$ we can see 
that either none or two can be negative, while $t_3$ can be either 
positive or negative. 
Because we are interested in channels obtained by concatenation, and because 
we always allow for any number of unitary conjugations, we can always remove
any negative signs in $\lambda_j$  
or $t_3$ by unitary channels, as the following lemma shows.
\begin{lemma}
Any extremal qubit channel $\Phi$, written in the canonical form of 
Eq.(\ref{eq:extremal}), can be transformed by unitary conjugations into an extremal channel $\Phi'$ of 
the same form, but with all $\lambda_j$ and $t_3$ non-negative. 
\end{lemma}
\begin{proof}
Let $\Phi_U$ denote a channel corresponding to a unitary
 conjugation. Unitary conjugations can change the sign of arbitrary two
 $\lambda_j$. First, using the composition rule in
 Eqs.~(\ref{eq:compose-M}),(\ref{eq:compose-t}), we observe that by composing
 $\Phi'=\Phi_U \circ \Phi$ we can change the sign of $t_3$ and the
 sign of $\lambda_3$ and either $\lambda_1$ or $\lambda_2$. On the
 other hand, with $\Phi'=\Phi \circ \Phi_U$ we can change the sign of
 any two $\lambda_j$ while leaving the sign of $t_3$ intact. Combining
 concatenations with a unitary before and after the channel $\Phi$, we
 can change any allowed combination of signs. Let us discuss all
 possible cases: (i) if $t_3\ge 0$ and any two $\lambda_j$ are
 negative we simply change the sign of these two $\lambda_j$ by
 $\Phi'=\Phi \circ \Phi_U$; (ii) if $t_3<0$ and $\lambda_3<0$ as well
 as one of $\lambda_{1,2}$ is negative we can change all signs at once
 with $\Phi'=\Phi_U \circ \Phi$; (iii) if $t_3<0$ and
 $\lambda_{1,2}<0$, or $t_3<0$ and all $\lambda_j$ are non-negative we can,
 using $\Phi'=\Phi 
 \circ \Phi_U$, bring the channel to the form discussed under (ii).    
\end{proof}

We therefore see that with unitary conjugations we can always bring an
extremal qubit channel to the form given by Eq.~(\ref{eq:extremal})
with all $\lambda_j$ and $t_3$ being non-negative, that is to the set
with angles $u \in (0,\pi/2]$ and $v \in (0,\pi/2]$ (we must also have
$u\le v$ due to $|\lambda_1| \ge |\lambda_2|$), plus an additional
point $u=v=0$. From now on we shall limit our discussion to that
subset of extremal qubit channels. They can be further classified
according to the number of different pure output points one gets for
different values of $u$ and $v$. In addition, it will be useful to
classify channels also depending on whether they represent an
injective map. 

\begin{definition}
A channel $\Phi$ is called degenerate iff the determinant of $T_\Phi$ (\ref{T}) is zero; otherwise it is called non-degenerate. Equivalently, a channel is non-degenerate iff all $\lambda_j$ are non-zero, i.e., iff the volume of the set of output states is nonzero.
\end{definition}

Depending on the degeneracy and the number of output pure states, extremal
qubit channels (\ref{eq:extremal}) can be classified as follows. 

\begin{lemma}[\cite{Ruskai02}]\label{lem:class}
The set of extremal qubit channels ${\cal X}$ can be classified according to
the number of pure outputs as $\mathcal X={\cal U}_2 \cup {\cal X}_{\rm
  1PO}^{\rm deg} \cup {\cal X}_{\rm 1PO}^{\rm nd}\cup {\cal X}_{\rm
  2PO}^{\rm deg} \cup {\cal X}_{\rm 2PO}^{\rm nd}$ where the pure outputs of
the subclasses are as follows: 
\begin{enumerate}
\item Unitary conjugations ${\cal U}_2$ for which $\mathrm{PO}(\Phi) =
  \mathcal P_2$. They correspond to $u=v=0$, or $\omega=\theta\ne\pi/2$. 
\item $\mathrm{PO}(\Phi) = \{\xi\}$: the set given by $u=v$ that can be
  conveniently parametrized as $\bt_\Phi=(0,0,1-\lambda^2)$ and
  $\bm{\lambda}_\Phi=(\lambda,\lambda,\lambda^2)$. Note that the mapping
  from $u,v$ to $\theta,\omega$ is not injective in this case and all such
  channels correspond to a single point $\theta=\omega=\pi/2$ in the
  $\theta-\omega$ plane (Fig.~\ref{fig:extremal}).  
\begin{enumerate}
\item Degenerate channels ${\cal X}_{\rm 1PO}^{\rm deg}$: $\lambda=0$, i.e., $u=v=\pi/2$. The set of output states is a single point on the Bloch sphere.
\label{class:1POdeg}
\item Non-degenerate channels ${\cal X}_{\rm 1PO}^{\rm nd}$: the set given by $\lambda \in (0,1)$, i.e., $u=v$ and $u\in (0,\pi/2)$.
\label{class:1POnd}
\end{enumerate}
\item $\mathrm{PO}(\Phi) = \{\xi,\zeta \}$, with $\xi \neq \zeta$:
\begin{enumerate}
\item Degenerate channels ${\cal X}_{\rm 2PO}^{\rm deg}$: the set given by $v=\pi/2$ and $u \in (0,\pi/2)$, i.e., $\theta=0$ and $\omega \in (0,\pi/2)$. The set of output states is a degenerate ellipsoid -- a line segment. 
\label{class:2POdeg}
\item Non-degenerate channels ${\cal X}_{\rm 2PO}^{\rm nd}$: the set
  given by $0 < \theta < \omega < \pi/2$, or, in angles $u,v$, the set
  $0<u<v<\pi/2$ ($u=0$ with $v >0$ is excluded). This is the interior of
  the shaded triangle in Fig.~\ref{fig:extremal}.
\label{class:2POnd}
\end{enumerate}
\end{enumerate}
\end{lemma}

We shall investigate these classes of extremal channels and show 
that any universal set of qubit quantum channels needs to contain 
some of these maps. We analyze each individual case in the next two 
subsections. Before that, let us state few general statements that 
will be useful in proving decompositions.

\begin{lemma}\label{lem:mixed}
Let $\Phi : \M_d(\C) \to \M_d(\C)$ be a quantum channel with the 
property that there exists a full rank input state $\rho$ such that 
$\Phi(\rho) = |\psi \rangle \langle \psi|$, a rank-one projector. Then, 
the channel $\Phi$ is constant, i.e. for all input states $\sigma$, 
$\Phi(\sigma) = |\psi \rangle \langle \psi|$.
\end{lemma}
\begin{proof}
Let $\sigma$ be any input state. Since $\rho$ is full rank, there 
exists a positive constant $\varepsilon$ such that $\varepsilon \sigma 
\leq \rho$. Being a quantum channel, the map $\Phi$ preserves 
positivity, hence $\varepsilon \Phi(\sigma) \leq \Phi(\rho) = |\psi 
\rangle \langle \psi|$. This means that $\Phi(\sigma)$ has support
only in the $|\psi\rangle\langle\psi|$ subspace, i.e.~$\Phi(\sigma)=c
|\psi\rangle\langle\psi|$ with some $c>0$. Using trace 
preservation, we conclude that $c=1$  and thus 
$\Phi(\sigma) = |\psi \rangle \langle \psi|$.
\end{proof}

Note that the full-rank hypothesis in the above lemma is necessary; taking a
direct sum of two constant channels shows that a mixed input is not enough
to guarantee that the channel is constant. However, for qubits, the notions
of mixed state and full-rank state are equivalent. Among qubit channels such
constant channels are exactly channels of the type \ref{class:1POdeg} in
Lemma~\ref{lem:class}.

\begin{lemma}\label{lem:nondeg}
In a decomposition of a non-degenerate channel there can be only non-degenerate 
channels.
\end{lemma}
\begin{proof}
Writing $\Phi=\Phi_n \circ \cdots \circ \Phi_1$ and taking the determinant on both sides of $T_{\Phi} = T_{\Phi_n} \cdots T_{\Phi_1}$, we see that the determinant of the right side can be non-zero (i.e., $\Phi$ is non-degenerate) only if all channels $\Phi_j$ are non-degenerate.
\end{proof}

\begin{proposition}\label{prop:extremal}
In a decomposition of a non-degenerate extremal channel $\Phi$ there can be only non-degenerate extremal channels.
\end{proposition}
\begin{proof}
Consider a decomposition $\Phi=\Phi_1 \circ \Phi_2$. Due to Lemma~\ref{lem:nondeg} we know that both $\Phi_1$ and $\Phi_2$ must be non-degenerate. Suppose $\Phi_2$ is not extremal, that is, we can write it as a nontrivial convex sum $\Phi_2=c_1 \Psi_1+c_2 \Psi_2$, with $\Psi_1 \neq \Psi_2$. Using this sum $\Phi$ can be written as $\Phi=c_1 \Phi_1 \circ \Psi_1+c_2 \Phi_1 \circ \Psi_2$. Because $\Phi$ is supposed to be extremal $\Phi_1 \circ \Psi_1$ must be equal to
$\Phi_1 \circ \Psi_2$, otherwise $\Phi$ would be a nontrivial convex 
combination. But because $\Phi_1$ is non-degenerate, i.e., an
injection, and $\Psi_1 \neq \Psi_2$, there is at least one point whose
image under $\Psi_1$ is different from its image  under $\Psi_2$, and
therefore 
$\Phi_1 \circ \Psi_2$ can not be equal  
to $\Phi_1 \circ \Psi_1$. $\Phi_2$ must therefore be extremal. For the case when
$\Phi_1$ would be a convex combination, the argument is analogous. Therefore, 
neither $\Phi_1$ nor $\Phi_2$ can have a nontrivial convex combination.
\end{proof}

Finally, let us make a  remark regarding the relation between extremal
and indivisible channels. Somewhat unintuitively, all extremal channels
are divisible. This follows from the characterization of extremal 
\cite{Ruskai02} and indivisible \cite[Theorem 23]{Wolf08} channels. Indeed, indivisible channels are unital and this implies
$t_3= 0$ in Proposition \ref{prop:extremal-channels}. This, in turn, implies
$u=0$, which is an excluded parameter. Most notably, the indivisible
channel $\rho \mapsto (\rho^t + (\tr\rho) I_2)/3$ from \cite{Wolf08} is not extremal,
as it is unital with $\bm{\lambda}=(1/3,-1/3,1/3)$, which corresponds to the
center of a face of the tetrahedron.

\subsection{Extremal qubit channels with one pure output}

\subsubsection{Degenerate channels }

Consider a \emph{generalized depolarizing} (or \emph{constant}) channel
$\mathcal Q_\rho\in\mathcal C_d$ defined 
as 
\begin{equation}
	\mathcal Q_{\rho_0}(\rho) = (\tr \rho)
	\rho_0, 
\end{equation}
where $\rho_0\in \mathcal D_d$ is a fixed density operator. The usual
depolarizing channel is a particular case obtained by considering for
$\rho_0$ the
maximally mixed state $I_d / d$. An important feature of
generalized depolarizing channels is that their image (as quantum channels)
is trivial: $\mathcal Q_{\rho_0}(\mathcal D_d) = \{\rho_0\}$, that is, all
states are mapped onto a single point.

We now look at generalized depolarizing qubit channels ($d=2$). If ${\bf r}$ is
the Bloch vector of the state $\rho_0 \in \mathcal D_2$, then $M_{{\mathcal
    Q}_{\rho_0}}= {\rm diag}(0,0,0)$ and $\bt_{{\mathcal Q}_{\rho_0}} = {\bf
  r}$. Of special interest to us are channels with pure $\rho_0$, that is,
extremal channels of the form \ref{class:1POdeg} in Lemma \ref{lem:class}.   

\begin{proposition}\label{prop:decomposition-1PO-deg}
For a pure state $|\psi \rangle \langle \psi| \in \mathcal P_2$, consider a decomposition
\begin{equation}
	\mathcal Q_{|\psi \rangle \langle \psi|} = \Phi_2 \circ \Phi_1.
\end{equation}
Then, at least one of $\Phi_{1,2}$ is a constant channel
$\mathcal Q_{|\phi \rangle \langle \phi|}$, for some pure state $|\phi \rangle \langle \phi| \in \mathcal P_2$.  
\end{proposition}
\begin{proof}
Let $\rho$ be an arbitrary mixed (and thus full-rank) input state for the channel 	$\mathcal Q_{|\psi \rangle \langle \psi|}$, 
\begin{equation}
	\Phi_2(\Phi_1(\rho)) = |\psi \rangle \langle \psi|,
\end{equation}
and consider the intermediary state $\sigma = \Phi_1(\rho)$. If $\sigma$ is a pure state, then the channel $\Phi_1$ satisfies the hypothesis of Lemma \ref{lem:mixed}, and it is therefore constant. Otherwise, $\sigma$ is a mixed state, but then $\Phi_2(\sigma) =  |\psi \rangle \langle \psi|$, and, by the same Lemma \ref{lem:mixed}, $\Phi_2$ is constant. 
\end{proof}

\begin{corollary}\label{coroll:1POdeg}
Any set
of universal qubit channels contains at least one generalized
depolarizing channel $\mathcal 
Q_\rho$ for some pure state $\rho \in \mathcal P_2$.  As all other
generalized depolarizing channels can be obtained from $\mathcal Q_\rho$ by concatenation with some
unitary conjugation, it is also sufficient to have a single generalized
depolarizing channel in 
the universal set for the creation of {\em all}
generalized depolarizing qubit channels with pure output.
\end{corollary}

\subsubsection{Non-degenerate channels}

Let $\Phi_{\rm 1PO}(\lambda)$ be an extremal 1PO
channel  (type \ref{class:1POnd} in Lemma~\ref{lem:class}) with a
parametrization  
\begin{equation}\label{1POnd}
\bt_\Phi=(0,0,1-\lambda^2) \mbox{ and }
\bm{\lambda}_\Phi=(\lambda,\lambda,\lambda^2)\,. 
\end{equation}
We define
a set of length $\varepsilon$ by ${\cal X}_{\rm 1PO}^{\rm
  nd}(\varepsilon)$,  
\begin{equation}
{\cal X}_{\rm 1PO}^{\rm nd}(\varepsilon)=\{ \Phi_{\rm 1PO}(\lambda), \lambda \in (1-\varepsilon,1) \},
\end{equation}
where $\varepsilon$ is any positive number less than $1$.

\begin{lemma}\label{lem:1PO}
The set of all 1PO non-degenerate extremal channels ${\cal X}_{\rm 1PO}^{\rm nd}$ can be obtained by concatenation from the 1PO non-degenerate extremal universal set ${\cal X}_{\rm 1PO}^{\rm nd}(\varepsilon)$.
\end{lemma}
\begin{proof}
It is straightforward to check that the following concatenation rule holds for 1PO non-degenerate extremal channels, $\Phi_{\rm 1PO}(\lambda \mu)=\Phi_{\rm 1PO}(\lambda)\circ \Phi_{\rm 1PO}(\mu)$, $\lambda,\mu \in (0,1)$. Therefore, by a completely analogous argument as in the case of phase-flip channels, Lemma~\ref{lemtopedge}, we can see that concatenating at most $n$ channels from ${\cal X}_{\rm 1PO}^{\rm nd}(\varepsilon)$, where $\varepsilon=1-T^{1/n}$, we can get $\Phi_{\rm 1PO}(T)$, with any $T \in (0,1)$.
\end{proof}

The following proposition shows that any universal set must containt at
  least some element of the set ${\cal 
  X}_{\rm 1PO}^{\rm nd}(\varepsilon)$.
\begin{proposition}\label{prop:1PO}
Let $\Phi_{\rm 1PO}^{\rm nd}$ be a non-degenerate extremal channel with 1 PO
and $\Phi_{\rm 1PO}^{\rm nd}=\Phi_2 \circ \Phi_1$ an arbitrary
decomposition. Then, up to unitary conjugations, both $\Phi_{1,2}$ must be
1PO non-degenerate extremal channels. 
\end{proposition}
\begin{proof}
According to Proposition~\ref{prop:extremal} $\Phi_{1,2}$ can be
either unitaries, 1PO or 2PO non-degenerate extremal channels. A 1PO 
non-degenerate extremal channel in the parametrization (\ref{1POnd})
maps a pure input state with 
$\theta=\pi/2$, i.e.~the North Pole $\eta$ of the
Bloch sphere to itself, $\Phi_{\rm 1PO}^{\rm
  nd}(\eta)=\eta$. If $\Phi_1(\eta)$ 
is a mixed state, then $\Phi_2$ maps a mixed state to a pure output,
so, by Lemma \ref{lem:mixed}, it is a constant, hence degenerate
channel. This would contradict the non-degeneracy of
$\Phi_{1PO}^{\rm nd}$.
Therefore, we must have that $\Phi_1$ maps $\eta$ to some
pure state $\zeta$, $\Phi_1(\eta)=\zeta$, and then
$\Phi_2(\zeta)=\eta$. This shows that $\Phi_1$ and $\Phi_2$ have at
least one pure output. In order to exclude possible 2PO channels from 
the decomposition we shall use a local argument about the curvature of
the boundary of output sets at these pure output points.  

The output set of $\Phi_{\rm 1PO}^{\rm nd}$ as parametrized by
  (\ref{1POnd}) is an ellipsoid touching the
  Bloch sphere at the north pole. Any plane containing the north pole
  and the origin intersects this output ellipsoid in an ellipse with
  the major axis $a=\lambda$ and the minor axis $b=\lambda^2$. The
  radius of curvature of an ellipse at its vertices that are closest
  to the ellipse centre is $a^2/b$ and is therefore $R=1$ in our
  case. We also observe that for any extremal channel, i.e., an
  ellipsoid touching a sphere from inside, the radius of curvature at the
  touching point in
  any plane containing a pure output state is upper-bounded by the
  radius of curvature of the Bloch sphere (which is $1$). For a plane
  containing two pure output points and the origin, where the output
  set is an ellipse with a major axis $a=\cos{u}$ and a minor axis
  $b=\cos{u}\cos{v}$ (\ref{eq:extremal}), one can explicitly calculate
  that the radius of curvature at the touching point is
  $R=(\cos{v}/\cos{u})^2$. In particular, for unitary conjugations it
  is of 
  course $1$, whereas for a 2PO non-degenerate extremal channel it is
  always less than $1$. As the output set of the concatenation must be
  in the output set of $\Phi_2$, the curvatures of the ellipsoid at a
  PO can never decrease under concatenation, or, equivalently, the
  radius of curvature can never increase. 
Because the curvature of the final $\Phi_{\rm 1PO}^{\rm nd}$ must be
  $1$, we conclude that $\Phi_{1,2}$ can never be non-degenerate 2PO 
  extremal channels.        
\end{proof}

Combining Corollary~\ref{coroll:1POdeg}, Lemma~\ref{lem:1PO} and Proposition~\ref{prop:1PO}, we obtain an $\varepsilon$-small universal set for 1PO extremal channels.

\begin{corollary}\label{coroll:1PO}
For any $\varepsilon >0$, the following set is universal for 1PO extremal  channels:
\begin{equation}
{\cal X}_{\rm 1PO}(\varepsilon)=\{ \Phi_{\rm 1PO}(\lambda), \lambda \in (1-\varepsilon,1) \cup \{0\} \},
\label{eq:1POepsilon}
\end{equation}
where $\Phi_{\rm 1PO}(\lambda)$ is a channel with $\bt_\Phi=(0,0,1-\lambda^2)$ and $\bm{\lambda}_\Phi=(\lambda,\lambda,\lambda^2)$.
\end{corollary}

\subsection{Extremal qubit channels with two pure outputs}
Let  $\Phi_{1,2}$ be two extremal 2PO channels  of form
\eqref{eq:extremal} with parameters $(\omega_i, \theta_i)$ such that $\omega_1 = \theta_2$. Then, $\Phi = \Phi_2 \circ \Phi_1$ is of the form \eqref{eq:extremal}  with parameters $(\omega_2, \theta_1)$, that is
$\Phi(\omega_2,\theta_1)=\Phi_2(\omega_2,\omega_1) \circ
\Phi_1(\omega_1,\theta_1)$. Such channels, with parameters  $(\omega,
\theta)$ will be denoted by $\Phi_{\rm 2PO}(\omega,\theta)$.

\subsubsection{Degenerate channels}
2PO Degenerate extremal channels map two orthogonal pure input states
($\theta=0$) to two pure output states with $0<\omega<\pi/2$ (type
\ref{class:2POdeg} in Lemma~\ref{lem:class}). The set 
of output states is a line segment touching the Bloch sphere in two
points with an angle $\omega$.  

\begin{lemma}\label{lem:2PO}
For any $\varepsilon >0$, the set of all 2PO degenerate extremal channels ${\cal X}_{\rm 2PO}^{\rm deg}$ can be obtained by concatenation of channels from the set
\begin{equation}
{\cal X}_{\rm 2PO}^{\rm deg}(\varepsilon)=\{\Phi_{\rm 2PO}(\omega,0), \omega \in (0,\varepsilon) \},
\label{eq:2POdegepsilon}
\end{equation}
and channels from the set ${\cal X}_{\rm 2PO}^{\rm nd}$. 
\end{lemma}
\begin{proof}
The statement immediately follows from a general composition rule for two genuine 2PO extremal channels (those with non-equal 2 POs) saying that $\Phi_{\rm 2PO}(\omega,0)=\Phi_{\rm 2PO}(\omega,x) \circ \Phi_{\rm 2PO}(x,0)$, with any $x\in (0,\omega)$, and the fact that ${\cal X}_{\rm 2PO}^{\rm nd}$ contains all $\Phi_{\rm 2PO}(\omega,x)$ with $0<x<\omega<\pi/2$.
\end{proof}

\begin{lemma}\label{lem:2podeg}
If a qubit channel $\Phi$ maps two orthogonal pure states to two distinct non-orthogonal pure states, it must be a 2PO degenerate extremal channel, i.e., of the type \ref{class:2POdeg} in Lemma~\ref{lem:class}.
\end{lemma}
\begin{proof}
Let us denote two orthogonal pure input states by $\xi_a$ and $\xi_b$, and 
their pure output states by $\xi_a'$ and $\xi_b'$, $\Phi(\xi_a)=\xi_a',\,
\Phi(\xi_b)=\xi_b'$. $\Phi$ is an affine map and so it maps a line
segment $\overline{\xi_a\xi_b}$ to a line segment
$\overline{\xi_a'\xi_b'}$. The midpoint  
of $\overline{\xi_a\xi_b}$, which is a centre of the Bloch sphere, 
is mapped to a midpoint of $\overline{\xi_a'\xi_b'}$, which is the 
centre of the output ellipsoid. Let us denote by ${\cal P}$ a plane 
containing the ellipsoid centre and the points $\xi_a'$ and $\xi_b'$.
 In the plane ${\cal P}$ the output ellipsoid is an ellipse touching a 
circle from inside in points $\xi_a'$ and $\xi_b'$. Due to 
Lemma~\ref{euclid} we know that the large ellipse axis is parallel 
to $\overline{\xi_a'\xi_b'}$ and, because it must also pass through the 
midpoint of $\overline{\xi_a'\xi_b'}$, we conclude that the large axis must 
be equal to $\overline{\xi_a'\xi_b'}$. If the ellipse is supposed to only 
touch the circle and not intersect it, its small axis must be zero.  
The plane ${\cal P}$ therefore intersects the ellipsoid in a line segment of 
nonzero length. Orienting the coordinate system so that the $x$ axis is parallel 
to $\overline{\xi_a'\xi_b'}$, while the $z$ axis is in the plane ${\cal P}$ 
and perpendicular to $\overline{\xi_a'\xi_b'}$, 
we have $\lambda_3=0$ as well as $t_2=t_1=0$. Because also 
$t_3^2+\lambda_1^2=1$, we see that the coefficient $b$, 
Eq.~(\ref{coeffb}) in the GFA Theorem~\ref{thm:GFA}, is equal to 
$b=1-\lambda_1^2-\lambda_2^2-t_3^2=-\lambda_2^2$, which is non-negative 
only if $\lambda_2=0$. The channel $\Phi$ is therefore of the form
\ref{class:2POdeg}  
in Lemma~\ref{lem:class}.
\end{proof}

Using that Lemma we can now show that 2PO degenerate extremal channels are 
also necessary.

\begin{proposition}\label{prop:2POdeg}
Let $\Phi=\Phi_2 \circ \Phi_1$ be an arbitrary decomposition of a
2PO degenerate extremal channel $\Phi$. Then exactly one of the
channels $\Phi_{1,2}$ must be a 2PO degenerate extremal channel. 
\end{proposition}
\begin{proof}
A 2PO degenerate  extremal channel $\Phi$ maps two pure orthogonal states
$\xi_a$ and $\xi_b$ to  
two pure non-orthogonal states $\xi_a'$ and 
$\xi_b'$, $\Phi(\xi_a)=\xi_a',\, \Phi(\xi_b)=\xi_b'$.  
Due to Lemma~\ref{lem:mixed}, and because the states $\xi_a'$ and
$\xi_b'$  
are distinct, we know that neither of $\Phi_{1,2}$ can be a 1PO degenerate 
extremal channel, i.e., a channel that would map a mixed state to a 
pure state. Therefore, the image of pure states $\xi_a$ and $\xi_b$
under 
$\Phi_1$ must be two pure states, say $\zeta_a$ and $\zeta_b$,  
$\Phi_1(\xi_a)=\zeta_a$, $\Phi_1(\xi_b)=\zeta_b$. If $\zeta_a$ and
$\zeta_b$ 
are non-orthogonal, then $\Phi_1$ maps two orthogonal states $\xi_a$ and
$\xi_b$ to two non-orthogonal pure states and, according to
Lemma~\ref{lem:2podeg}, must be a 2PO degenerate extremal channel,
whereas $\Phi_2$ cannot be a 2PO degenerate extremal channel as it
maps two non-orthogonal states to two non-orthogonal states. If
on the other hand $\zeta_a$ and $\zeta_b$ are orthogonal, then
$\Phi_2$ must in turn map these two orthogonal  
pure states to two non-orthogonal pure states and must
therefore be a 2PO degenerate extremal channel, whereas $\Phi_1$ maps two
orthogonal pure states to two orthogonal pure states, and is therefore a 
non-extremal unital channel. 
\end{proof}

\subsubsection{Non-degenerate channels } 
According to Lemma~\ref{lem:class}, a channel $\Phi$ belonging to
class  \ref{class:2POnd} takes the form \eqref{eq:extremal} with
parameters $\theta$ and $\omega$ such that $0<\theta<\omega<\pi/2$
(shaded triangle in Fig.~\ref{fig:extremal}). 
>From Proposition
\ref{prop:extremal}, up to unitary conjugation, such a channel can only be
decomposed into 
channels of the same form. We obtain in
this way a universal set for 2PO non-degenerate channels. 

\begin{proposition}\label{prop:2PO-non-deg}
For any $\varepsilon>0$, the set $\{\Phi_{\rm 2PO}(\theta,\omega);
0<\omega-\theta<\varepsilon\}$ is a universal set for
2PO-non-degenerate channels. Moreover, given any non-trivial
decomposition of a 2PO-non-degenerate channel $\Phi = \Phi_2 \circ
\Phi_1$, both channels $\Phi_{1,2}$ must be extremal
2PO-non-degenerate. 
\end{proposition}
\begin{proof}
Any concatenation $\Phi=\Phi_n \circ\cdots\circ\Phi_2 \circ \Phi_1$ of
maps with parameters $(\theta_i,\omega_i)$ such that
$\omega_{i}=\theta_{i+1}$ is of the form \eqref{eq:extremal} with
parameters $(\theta_1,\omega_n)$. Concatenating $n$ channels with
parameters lying in the strip  $\{(\theta,\omega);
0<\omega-\theta<\varepsilon\}$ (blue area in Fig.~\ref{fig:extremal})
allows to reach any final angle
$\omega_n\in[\theta,\theta+n\varepsilon)$ from an initial angle
  $\theta$. Therefore any channel with parameters $(\theta,\omega)$
  can be decomposed into a sequence of channels with $n=\lfloor
  (\omega-\theta)/\varepsilon \rfloor+1$. 

The second statement follows from Proposition \ref{prop:extremal} and
the fact that both $\Phi_2$ and $\Phi_1$ should have exactly 2 pure
outputs. 
\end{proof}

Combining Lemma~\ref{lem:2PO} and Proposition~\ref{prop:2PO-non-deg}, we obtain an $\varepsilon$-small universal set for 2PO extremal channels.

\begin{corollary}\label{coroll:2PO}
For any $\varepsilon >0$, the following set is universal for extremal 2 PO channels:
\begin{equation}
{\cal X}_{\rm 2PO}(\varepsilon)=\{ \Phi_{\rm 2PO}(\omega,0), \omega \in (0,\varepsilon)\} \cup \{ \Phi_{\rm 2PO}(\omega,\theta), \omega-\theta \in (0,\varepsilon)\},
\label{eq:2POepsilon}
\end{equation}
where $\Phi_{\rm 2PO}(\omega, \theta)$ is an extremal channel mapping pure input states $(\pm \cos{\theta},0,\sin{\theta})$ to pure output states $(\pm \cos{\omega},0,\sin{\omega})$.
\end{corollary}

\subsection{A universal set for extremal qubit channels}

Finally, we state our main theorem, which is a compilation of Corollaries \ref{coroll:1PO} for 1PO extremal channels and \ref{coroll:2PO} for 2PO extremal channels. Note that, in virtue of Propositions \ref{prop:decomposition-1PO-deg}, \ref{prop:1PO}, \ref{prop:2POdeg} and \ref{prop:2PO-non-deg}, our results go beyond extremal channels, showing that any universal set of channels \emph{must} contain extremal channels belonging to the each class studied in this section. 

\begin{theorem}
For any $\varepsilon > 0$, the set ${\cal X}(\varepsilon)={\cal U}_2 \cup {\cal X}_{\rm 1PO}(\varepsilon) \cup {\cal X}_{\rm 2PO}(\varepsilon)$, where ${\cal X}_{\rm 1PO}(\varepsilon)$ is defined in \eqref{eq:1POepsilon} and ${\cal X}_{\rm 2PO}(\varepsilon)$ is defined in \eqref{eq:2POepsilon}, is a universal set for \emph{extremal qubit channels}. 

Moreover, any universal set of (general) qubit channels must contain
the following extremal channels: 
\begin{enumerate}
\item a 1PO degenerate (i.e. constant) channel $\mathcal Q_{| \psi \rangle \langle \psi |}$;
\item infinitely many 1PO non-degenerate extremal channels $\Phi_{\rm 1PO}(1-\varepsilon)$;
\item infinitely many 2PO degenerate extremal channels $\Phi_{\rm 2PO}(\varepsilon,0)$;
\item infinitely many 2PO non-degenerate extremal channels $\Phi_{\rm 2PO}(\omega, \theta)$, with $0<\omega-\theta<\varepsilon$.
\end{enumerate}
\end{theorem}  

\begin{figure}
\includegraphics[width=0.5\textwidth]{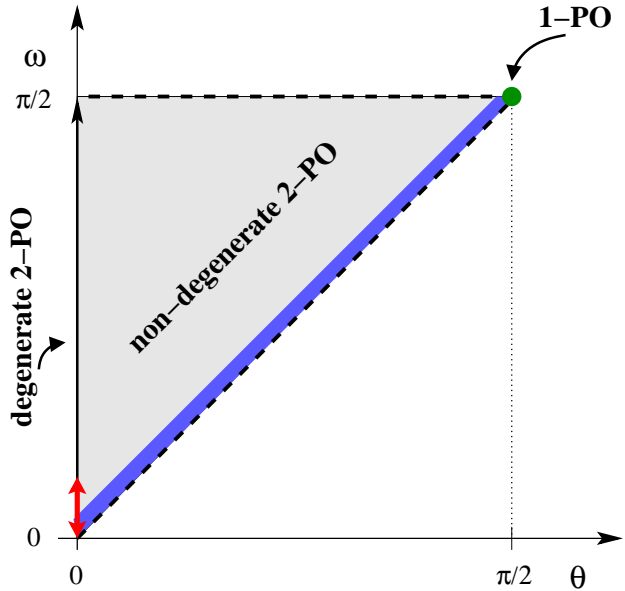}
\caption{Illustration of extremal qubit channels in the
  $(\theta,\omega)$ plane of latitude angles of the input/output pure states
  on the Bloch sphere. The
  inside of the shaded triangle  are non-degenerate extremal 2PO
  channels. The universal set for 2PO extremal channels is a union of
  an $\varepsilon$-strip above the diagonal (blue/dark color) and an
  $\varepsilon$-interval on the $\omega$-axis (red/bright color). 
  The 1PO extremal channels are in this parametrization represented by a single
  point at $\theta=\omega=\pi/2$. For details see classification in
  Lemma~\ref{lem:class}.}  
\label{fig:extremal}
\end{figure}

\section{Concluding remarks}
We have investigated the set of quantum channels acting on a single qubit,
i.e.~linear, trace preserving, and completely positive maps of the density
matrix.  We found a compact generalization of the
Fujiwara-Algoet conditions, i.e.~conditions for the complete
positivity of the map, to arbitrary (not necessarily unital) qubit
channels.    
  We used these conditions together with purely
geometrical considerations to examine the pure output
of the quantum channel.  We established that no qubit quantum channel
exists whose pure output is a circle of non-zero radius on the Bloch sphere,
generalizing the ``no-pancake theorem''.  We derived a universal set
of quantum 
channels for extremal qubit channels, i.e.~a set of quantum channels
from which all extremal qubit channels can be constructed by
concatenation. All other qubit channels can be constructed from these
extremal channels by simple classical
random sampling.  For
unital qubit channels we found a universal set of quantum channels
regardless of whether the qubit channel to be decomposed is extremal
or not.  
We showed that our universal sets are essentially minimal,
and must be contained in any universal set for
arbitrary (not necessarily extremal) qubit channels. 

\begin{acknowledgments}
C.~P.~ acknowledges financial support from the ANR project HAM-MARK,
N${}^\circ$ ANR-09-BLAN-0098-01.  
I.N.~'s research has been supported by the ANR projects {OSQPI} {2011 BS01
  008 01} and {RMTQIT}  {ANR-12-IS01-0001-01}, and by the PEPS-ICQ CNRS project \mbox{Cogit}.
\end{acknowledgments}

%\bibliography{../mybibs_bt}

\end{document}